\documentclass[10pt, conference]{./IEEEtran}

\usepackage{array, calc, color}
\usepackage{algorithm, algorithmic} %
\usepackage{graphics, epstopdf, graphicx, epsfig, psfrag, subfigure}
\usepackage{amssymb, amsmath, times}
\usepackage{multicol,balance, verbatim}
\usepackage{url, verbatim, xspace}
\usepackage[nospace]{cite}

\newcommand{\ie}{{\em i.e., }}
\newcommand{\eg}{{\em e.g., }}
\newcommand{\etal}{\emph{ et al. }}

\newcommand{\lastfm}{{\tt Last.fm}\xspace}

\newcommand{\fig}[1]{Fig.~\ref{#1}}
\newcommand{\eq}{\!\!=\!\!}

\newtheorem{theorem}{Theorem}[section]

\newtheorem{proposition}[theorem]{Proposition}

\newtheorem{definition}{Definition}

\makeatletter
\def\ps@headings{%
\def\@oddhead{\mbox{}\scriptsize\rightmark \hfil \thepage}%
\def\@evenhead{\scriptsize\thepage \hfil \leftmark\mbox{}}%
\def\@oddfoot{}%
\def\@evenfoot{}}
\makeatother
\begin{document}

\title{Multigraph Sampling of Online Social Networks}
\author{Minas Gjoka\\
CalIT2\\
UC Irvine\\
{\em mgjoka@uci.edu}\\
\and
Carter T. Butts\\
Sociology Dept\\
UC Irvine\\
{\em buttsc@uci.edu}\\
\and
Maciej Kurant\\
CalIT2\\
UC Irvine\\
{\em maciej.kurant@gmail.com}\\
\and
Athina Markopoulou\\
EECS Dept\\
UC Irvine\\
{\em athina@uci.edu}
}

\maketitle

\begin{abstract}
State-of-the-art techniques for probability sampling of users of online social networks (OSNs) are based on random walks on a single social relation (typically friendship).  While powerful, these methods rely on the social graph being fully connected. Furthermore, the mixing time of the sampling process strongly depends on the characteristics of this graph.

In this paper, we observe that there often exist other relations between OSN users, such as membership in the same group or participation in the same event. We propose to exploit the graphs these relations induce, by performing a random walk on 
their \emph{union multigraph}. 
We design a computationally efficient way to perform multigraph sampling by randomly selecting the graph on which to walk
at each iteration. %
We demonstrate the benefits of our approach through (i)~simulation in synthetic graphs, and (ii)~measurements of \lastfm - an Internet website for music with social networking features.
More specifically, we show that multigraph sampling can obtain a representative sample and faster convergence, even when the individual graphs fail, \ie are disconnected or highly clustered.
\end{abstract}


\section{Introduction}

The popularity of Online Social Networks (OSNs) has skyrocketed within the past decade, with the most popular having at present hundreds of millions of users (a number that continues to grow apace). This success has inspired a number of measurement and characterization studies, as well as studies of the interaction between OSN structure and systems design, and of user behavior within OSNs.  Despite their attractions, the large size and access limitations of most OSN services (\eg API query limits, treatment of user data as proprietary) make it difficult or impossible to obtain a complete census of user accounts and/or topology.  Sampling methods are thus essential for practical estimation of OSN properties.  While sampling can, in principle, allow precise inference from a relatively small number of observations, this depends critically on the ability to draw a sample with known statistical properties.
The lack of a sampling frame (\ie a complete list of users, from which individuals can be directly sampled) for most OSNs makes principled sampling especially difficult; recent work in this area has thus focused on sampling methods that evade this limitation.

Key to current sampling schemes is the fact that OSN users are, by definition, connected to one another via some relation, referred to here as the ``social graph.''  Specifically, samples of OSN users can be obtained by crawling the OSN social graph, obviating the need for a sampling frame.  An early family of crawling techniques followed BFS/Snowball-type approaches, where nodes of a graph reachable from an initial seed are explored exhaustively \cite{Ahn-WWW-07, Mislove2007, Wilson09}. It is now well-known that these techniques produce biased samples with poor statistical properties when the full graph is not
covered
\cite{Lee-Phys-Rev-06, Kurant2010, Gjoka2010}.  A more recent body of work employs systematic random walks on the social graph, and can achieve an asymptotic probability sample of users by online or a posteriori correction for the (known) bias induced by the crawling process \cite{Gjoka2010,Rasti2008}.  While random walk sampling can be very effective, its success is ultimately dependent on the connectivity of the underlying social graph.  More specifically, random walks can yield a representative sample of users only if the social graph is fully connected.  Furthermore, the speed with which the random walks converge to the target distribution strongly depends on characteristics of the graph, \eg clustering.

\begin{figure*}[t!]
\centering
\subfigure[Friendship graph]{
\includegraphics[scale=0.30, angle=0]{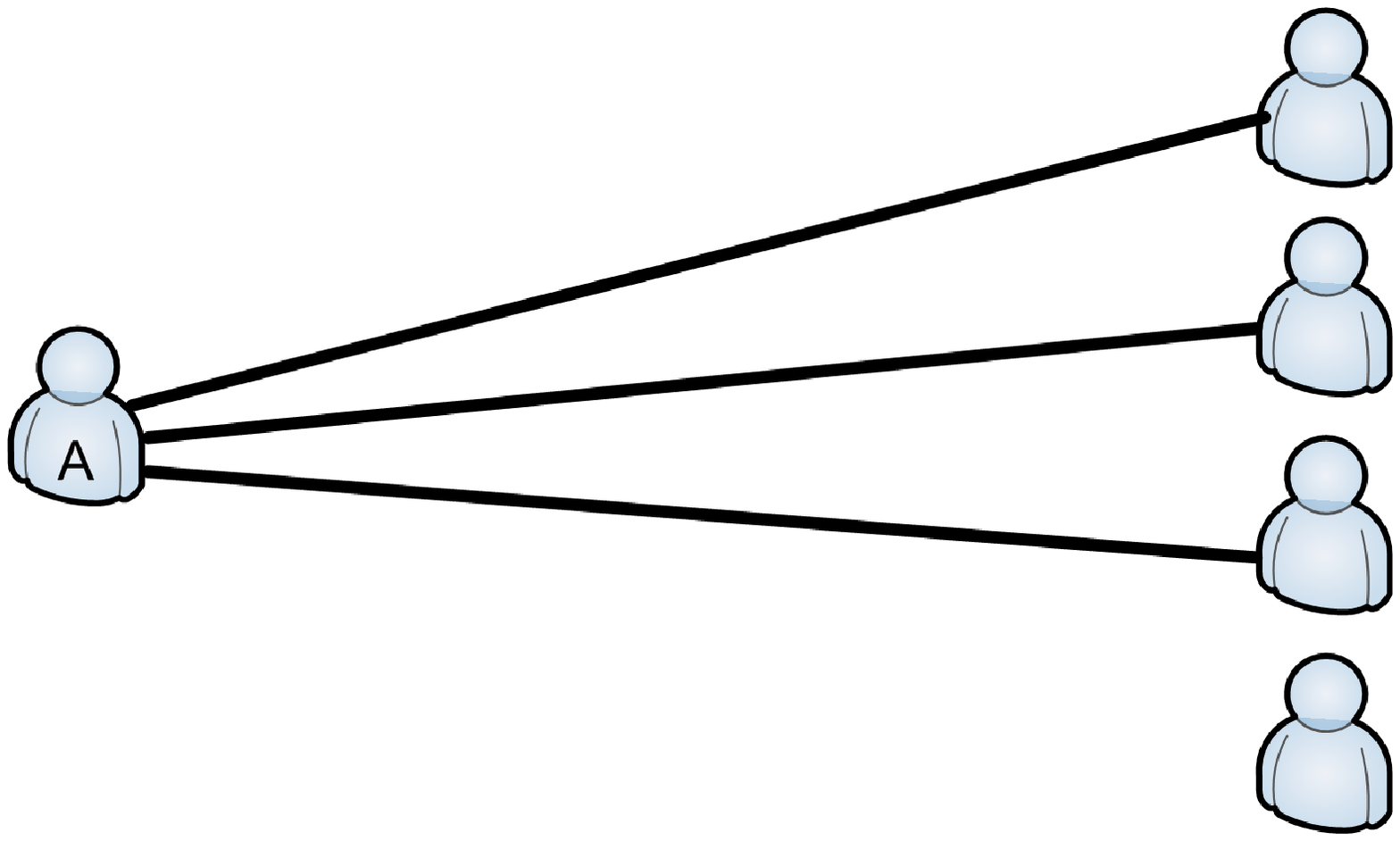}}\hspace{20pt}
\subfigure[Group graph]{
\includegraphics[scale=0.30, angle=0]{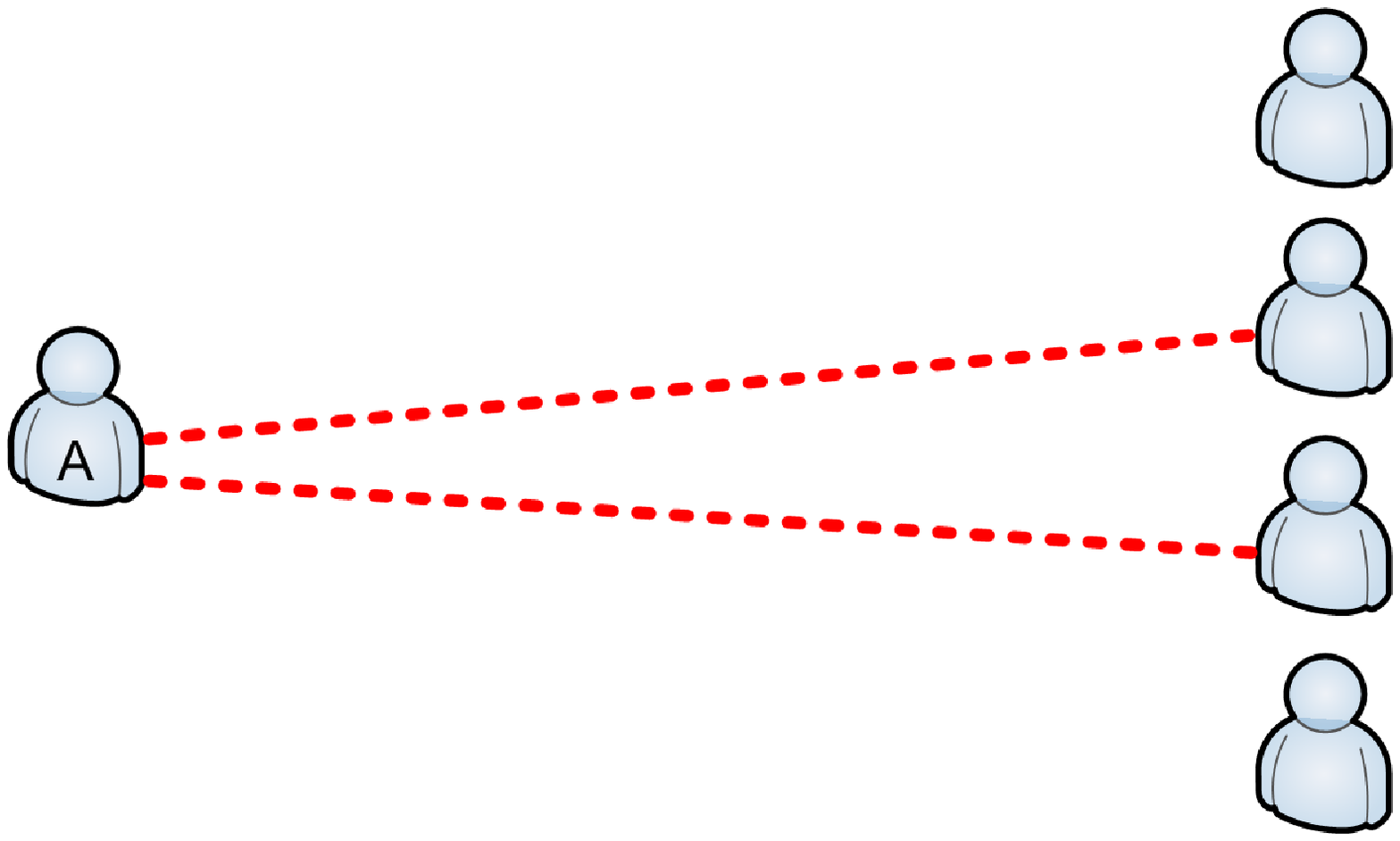}}\hspace{20pt}
\subfigure[Event graph]{
\includegraphics[scale=0.30, angle=0]{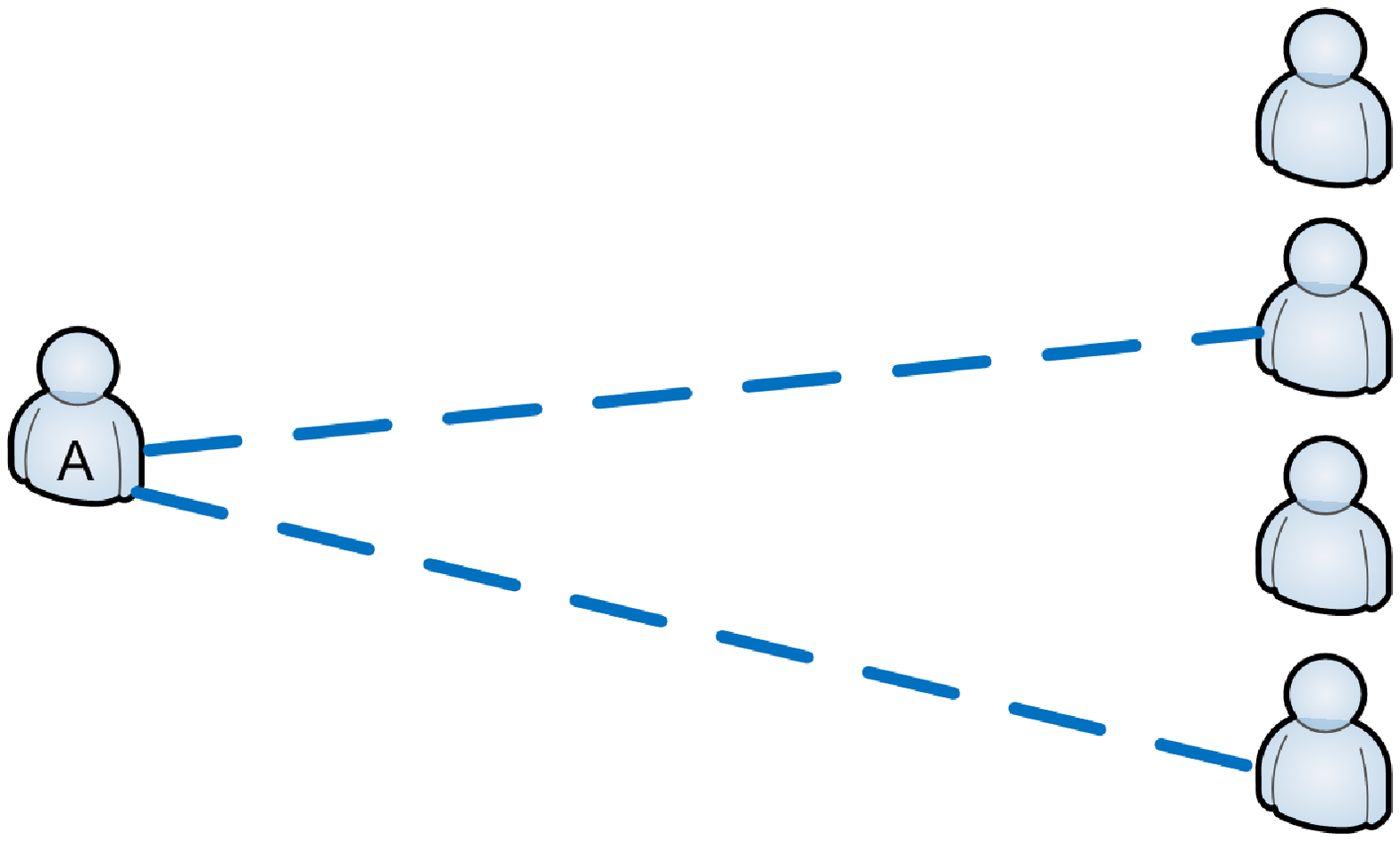}}\hspace{20pt}
\subfigure[Union simple graph]{
\includegraphics[scale=0.30, angle=0]{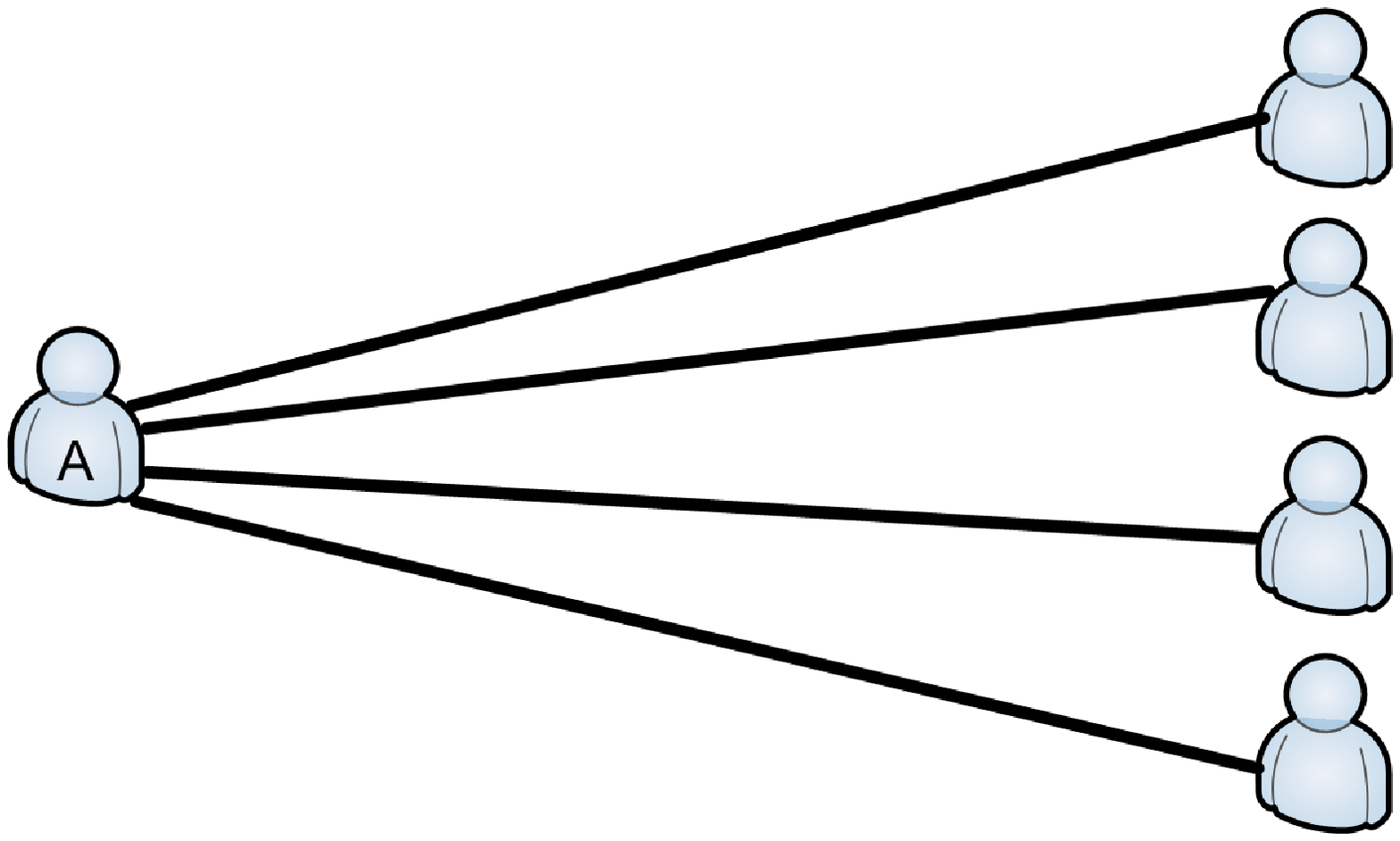}}\hspace{20pt}
\subfigure[The union multigraph contains {\em all} edges in the simple graphs]{
\includegraphics[scale=0.30, angle=0]{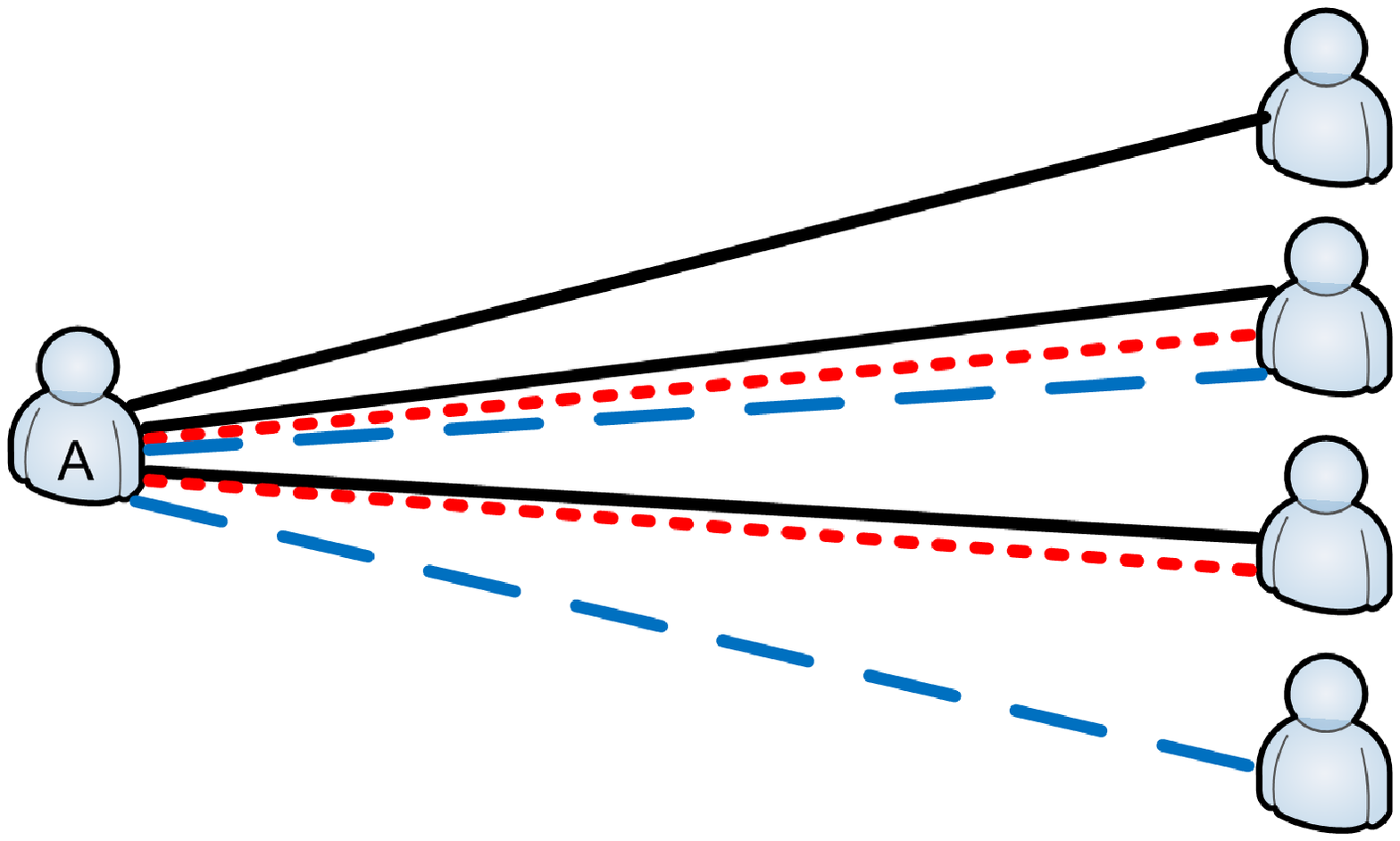}}\hspace{20pt}
\subfigure[An equivalent way of thinking the multigraph as ``mixture'' of simple graphs.]{
\includegraphics[scale=0.30, angle=0]{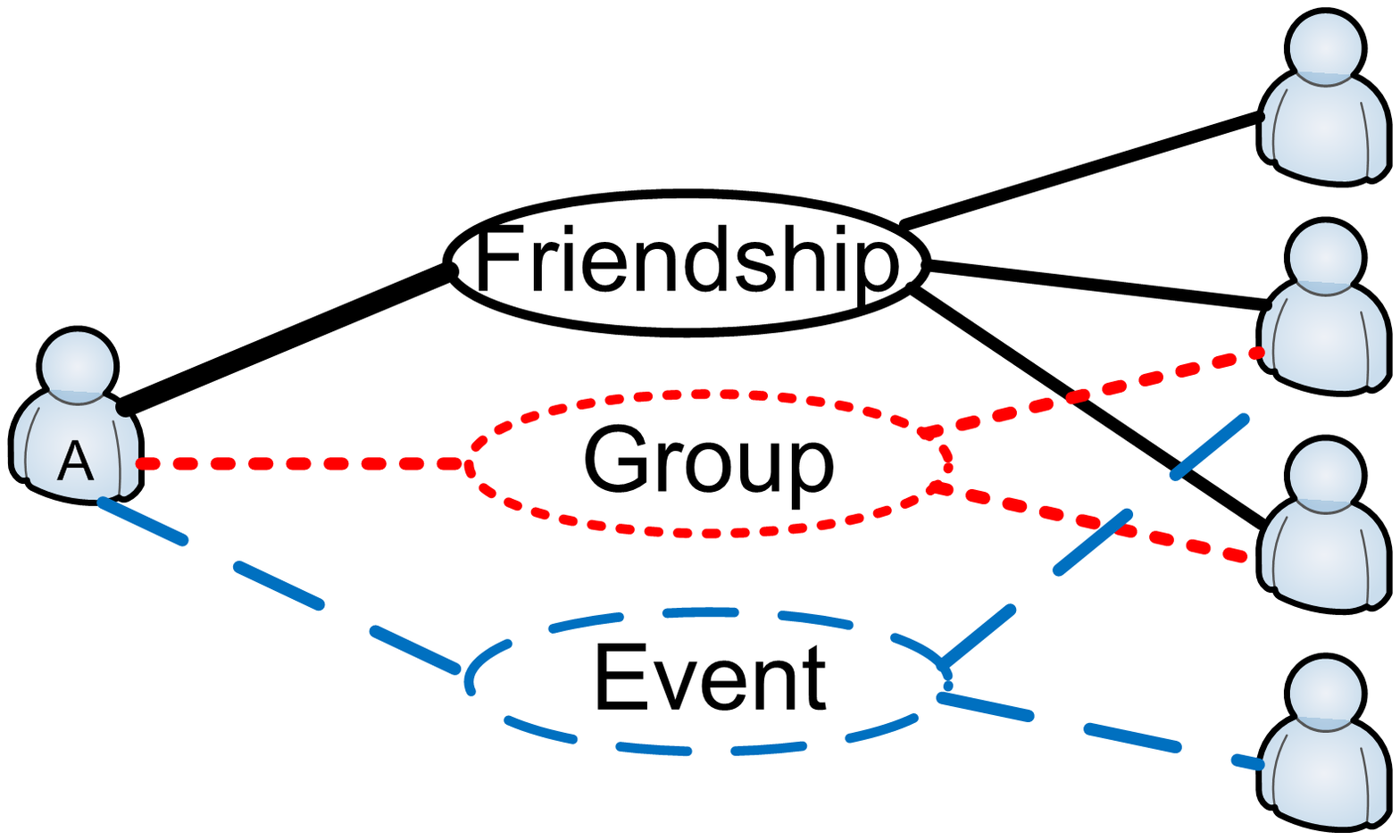}}
\caption{Multigraph sampling illustration. \textbf{(a-c)}~Graphs for three different relation~$G_i$: Friendship, Group and Event. \ \textbf{(d)}~Union (simple) graph, as presented in Definition~\ref{def:uniongraph}. \ \textbf{(e)}~Union multigraph, as presented in Definition~\ref{def:multigraph}. Node $A$ has degrees $d_1(A)\eq 3$, $d_2(A)\eq 2$ and  $d_3(A)\eq 2$ in the Friendship, Group and Event graphs, respectively. Its total degree in the union multigraph is %
$d(A)\eq 7$. \
\textbf{(f)}~An alternative view of the union multigraph. The neighbor selection Algorithm~\ref{alg:mixture}
first selects a graph~$G_i$ with probability $\frac{d_i(A)}{d(A)}$. Next, it picks a random neighbor in the selected graph~$G_i$, \ie with probability $\frac{1}{d_i(A)}$.
}
\label{fig:multigraph}
\end{figure*}

In this paper, we start from the observation that in OSNs, there are often multiple relations connecting the nodes. For example, users may be linked not only by direct social ties, but also by being members of the same group, participating in the same event, or using the same application. Moreover, many systems allow all neighbors in such relations to be enumerated (either through scraping or API calls). In other words, there often exist multiple, crawlable relation graphs---including but not limited to ties like ``friendship''---defined on the same set of nodes. We propose to exploit such multiple-relation (\ie multiplex) graphs by giving a crawler more edges to choose from,
compared to a crawler restricted to one relation only (typically the social graph). For example, we might be able to discover users that have no direct social ties, which is impossible by crawling the social graph alone.

There are many ways one can exploit multiplex graphs.
A naive approach would be to run many crawlers, one on each individual relation graph, and then combine the collected samples. However, this technique yields biased samples %
if any individual relation graph is fragmented, and fails to exploit opportunities for convergence acceleration by mixing across relations. A better approach is to combine all individual relation graphs into a single \emph{union} (simple) \emph{graph}: the resulting union graph is frequently connected even if its constituent graphs are not. Moreover, the union graph may also be less tightly clustered than its constituents, helping a crawler to converge faster than on the individual graphs.
However, walking on the union graph requires, at every step, the enumeration of all neighbors in all relations, which can be costly in time and bandwidth. Instead we propose a third, cost-efficient approach.

We propose a novel two-stage algorithm that walks on the \emph{union multigraph}. Our multigraph sampling first selects the relation on which to walk and then enumerates the neighbors with regards to that relation only, which makes it, in practice, even more efficient than union graph sampling. We prove that this algorithm achieves convergence to the proper equilibrium distribution when the union multigraph is connected.
We also demonstrate the benefits of multigraph sampling in two settings:
(i)~by simulation of synthetic random graphs; and
(ii)~by measurements of {\tt Last.fm} - an Internet website for music with social networking features.
We chose {\tt Last.fm} as an example of a network that is highly fragmented with respect to the social graph as well as other relations. We show that multigraph sampling can obtain a representative sample when each individual graph is disconnected.  Along the way, we also give practical guidelines on how to efficiently implement multigraph sampling for OSNs more generally.

The structure of the rest of the paper is as follows.  Section~\ref{sec:mg_samplingmethodology} describes our sampling methodology. %
Section~\ref{sec:mg_synthetic} evaluates our methodology on synthetic graphs. Section ~\ref{sec:mg_lastfm} applies our methodology to sample {\tt Last.fm} and provides practical recommendations. Section~\ref{sec:mg_background} discusses related work. 
Finally, Section~\ref{sec:mg_conclude} concludes the paper.

\section{Sampling Methodology\label{sec:mg_samplingmethodology}}

\subsection{Terminology and Definitions}

We consider different sets of edges $\mathcal{E}=\{E_1, \ldots, E_Q\}$ on a common set of users $V$.  Each $E_i$ captures a symmetric relation between users, such as friendship or group co-membership. $(V, E_i)$ thus defines an undirected graph $G_i$ on $V$.
We make no assumptions of connectivity or other special properties of each~$G_i$.
Fig~\ref{fig:multigraph}(a-c) shows an example of $Q\eq 3$ different relations and relation graphs~$G_i$ defined on the same 5 nodes.
Fig.~\ref{fig:bernoulli-example}(a-e) shows $Q\eq 5$ such graphs defined on a set of 50 nodes.

Consider set of graphs $G_i=(V, E_i), i=1,\ldots,Q$, defined on a common node set $V$. $\mathcal{E}$ can be used to construct several types of combined structures on $V$. We will employ the following two such structures:

\begin{definition}\label{def:uniongraph}
The {\em union (simple) graph} $G'=(V,E')$ of $G_1,\ldots,G_Q$ is defined as the graph on $V$,  whose edges are given by the set  $E'=\cup_{i=1}^Q E_i$. %
\hfill $\blacksquare$
\end{definition}

\begin{definition}\label{def:multigraph}
The {\em union multigraph} $G=(V,E)$ of $G_1,\ldots,G_Q$ is defined as the multigraph on $V$, whose edges are given by the multiset $E=\biguplus_{i=1}^Q E_i$. 
\hfill $\blacksquare$
\end{definition}

Note that the union multigraph $G$ can contain multiple edges between a pair of nodes, while the union graph $G'$ contains only one (or no) edge. 
Every multigraph  $G$ can be reduced to the union graph $G'$ by merging together multiple edges between two nodes into one.
Our focus, in this paper, is on the union multigraph, also referred to as simply the {\em multigraph}, because it allows us to more efficiently implement sampling on multiple relations. However, we also use the union graph as a helpful conceptual tool.

%
%
%

%
%
%
%
%

\subsection{Some False Starts}

We seek to draw a sample of the nodes in $V$, so that the draws are (at least approximately) independent and the sampling probability of each node is known up to a constant of proportionality.  There are several ways to achieve this goal using multiple graphs. We discuss some of them below.

\subsubsection{Naive Multiple Graph Sampling}

A naive way is to run many random walks, one per each individual graph~$G_i$, and to combine the collected samples.
However, if a particular~$G_i$ is disconnected (as are all five graphs in Fig.~\ref{fig:bernoulli-example}(a-e)),
a walk on $G_i$ is restricted to the connected component around its starting point and thus never converges to the desired target distribution.
This results in an asymptotically biased sample, dependent on the initial seeds.

\subsubsection{Union Simple Graph Sampling}\label{subsec:Union Simple Graph Sampling}
A much better approach is to perform a random walk on the union graph~$G'$.
We show examples of union graphs in Fig~\ref{fig:multigraph}(d) and Fig.~\ref{fig:bernoulli-example}(f). Note that although in both cases the individual graphs are disconnected, their union graphs are well-connected, which allows for a quick convergence of the random walk.

A potential practical difficulty with a random walk on the union graph~$G'$ is that, at each step, computing the neighborhood union can be quite expensive: it requires the enumeration of all edges adjacent to the current vertex~$v$, in each relation graph~$G_i$. This may be very costly, depending on $v$'s neighborhood size (which can be large in heavily clustered relations, such as group co-membership), query costs, and the number of relations~$Q$.

\begin{figure}[t]\label{ex:bernoulli-example}
\centering
\includegraphics[scale=0.32]{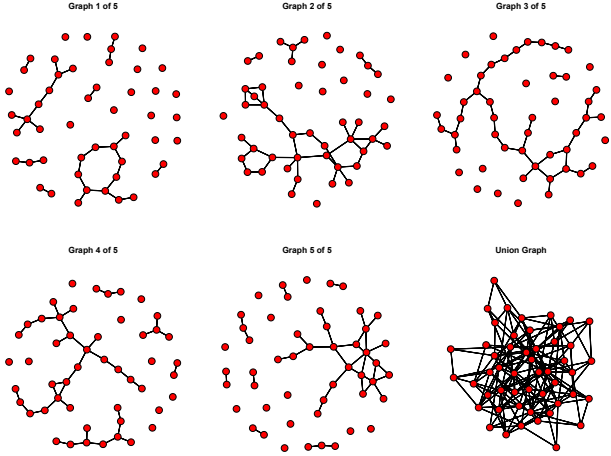}
\caption{Example of multiple graphs vs. their union. Five draws from a random ($N,p$) graph with $N=50$ and expected degree 1.5 are depicted in (a)-(e). Each simple graph is disconnected, while their union graph $G'$, depicted in (f), is fully connected. The union multigraph $G$ (not shown here) is also connected: it has - possibly multiple - edges between the exact same pairs of nodes as $G'$. }
\label{fig:bernoulli-example}
\end{figure}

\subsection{Union Multigraph Sampling}

We can address the enumeration problem of the union graph~$G'$ by considering the \emph{union multigraph}; see Definition~2 and example illustrated in Fig. \ref{fig:multigraph}(e).
We employ a random walk that moves from one vertex to another by selection of random edges on the multigraph.
A naive implementation of such a random walk still requires the enumeration of all neighbors of the current node~$v$.  Instead, we propose to use the following two-stage neighbor-selection procedure described in Algorithm~\ref{alg:mixture} and depicted in Fig.\ref{fig:multigraph}(f), which requires enumeration of $v$'s neighborhood for only a single graph.

Denote by $d_i(v)$ the degree of node $v$ in graph~$G_i$, and by  $d(v)=\sum_{i=1}^{Q} d_i(v)$ its total degree in the union multigraph.
First, we select a graph~$G_i$ with probability $\frac{d_i(v)}{d(v)}$.
Second, we pick uniformly at random an edge of $v$ within the selected~$G_i$ (\ie with prob. $\frac{1}{d_i(v)}$), and we follow this edge to $v$'s neighbor.
This procedure is equivalent to selecting an edge of $v$ uniformly at random in the union multigraph, because  $\frac{d_i(v)}{d(v)} \cdot \frac{1}{d_i(v)} = \frac{1}{d(v)}$.

Note that in Step~3, Algorithm~\ref{alg:mixture} requires only the values of degrees~$d_i(v)$ of all relation graphs.
Only in Step~4 of Algorithm~\ref{alg:mixture} does one enumerate all neighboring edges in the selected~$G_i$.
Because the degree information~$d_i(v)$ is usually much cheaper to obtain (\eg via simple low-bandwidth API calls)
than enumerating all~$d_i(v)$ edges, Algorithm~\ref{alg:mixture} has the potential to save much bandwidth compared to the union simple graph sampling (which enumerates all neighboring edges in all relations).
This benefit is amplified when higher numbers of relations ($Q$) are used.
Algorithm~\ref{alg:mixture} may also be helpful in certain offline applications involving surveys and human respondents (\eg RDS \cite{Heckathorn02_RDS}), in which selection of random neighbors is possible but enumeration is not.

\begin{algorithm}[t]
\caption{Multigraph Sampling Algorithm}
\begin{algorithmic}[1]
\REQUIRE {$v_0\in V$, simple graphs $G_i, i=1\ldots Q$}
\STATE {Initialize $v \leftarrow v_0$.}
\WHILE{\textsc{not Converged}}
\STATE {Select graph $G_i$ with probability $\frac{d_i(v)}{d(v)}$}
\STATE {Select uniformly at random a neighbor $v'$ of $v$ in $G_i$}
\STATE {$v \leftarrow v'$}
\ENDWHILE
\RETURN{all sampled nodes $v$ and their degrees $d(v)$.}
\end{algorithmic}
\label{alg:mixture}
\end{algorithm}

Algorithm~\ref{alg:mixture}
leads to the following equilibrium distribution:

\smallskip

\begin{proposition}
If $G$ is connected and contains at least one triangle, then Algorithm~\ref{alg:mixture} leads to equilibrium distribution $\pi(v)=\frac{d(v)}{\sum_{u\in V} d(u)}$.
\label{lemma:equilibrium}
\end{proposition}
\begin{proof}
Let $d(v,u)=d(u,v)$ be the number of edges between nodes $v$ and $u$ in the union multigraph $G$.
The sampling process of Algorithm~\ref{alg:mixture} is a Markov chain on $V$ with transition probabilities $P_{vu}=\frac{d(v,u)}{d(v)}, u,v \in V$.  So long as $G$ is finite and connected,
this random walk is irreducible and positive recurrent. The presence of a triangle within $G$ further guarantees aperiodicity.

A Markov chain of this type is equivalent to a random walk on an undirected weighted graph with edge weights $w(v,u)=d(v,u)$.
A random walk on weighted graph is known to have the unique equilibrium distribution
$\pi(v)=\frac{w(v)}{\sum_{u\in V} w(u)}$, where  $w(v)=\sum_{u} w(v,u)$
(\eg see \cite[Example 4.32]{RossBook}), and the proof follows immediately by substitution.
\end{proof}

\subsection{Practical Issues}

Various practical issues need to be addressed when implementing these ideas in practice. For completeness, we briefly repeat some good practices here, and we refer the interested reader
to our parallel work \cite{Gjoka2010, Gjoka2011_Facebook_JSAC}.

\subsubsection{Choice of crawling technique}

There are many ways to crawl a multigraph, \eg by using various random walks or graph traversal techniques. In \cite{Gjoka2011_Facebook_JSAC}, we showed that a simple random walk with correction for unequal sample weights, also called Re-Weighted Random Walk (RWRW),
is more efficient than competitors such as Metropolis-Hastings Random Walks. %
Therefore, throughout this paper, we employ RWRW described below.

\subsubsection{Re-Weighted Random Walk}\label{subsec:Re-Weighting the Random Walk}

Much like a classic random walk on a simple graph,
a random walk on multigraph~$G$ is inherently biased towards high-degree nodes. Indeed, per Proposition~\ref{lemma:equilibrium}, the probability of sampling a node~$v$ is proportional to its degree~$d(v)$ in~$G$.
\cite{Heckathorn97_RDS_introduction,Salganik2004,Newman01_EgoCentered_Networks} show how to apply the Hansen-Hurwitz estimator \cite{HansenHurwitz1943} to correct for this bias.
Let~$x(v)$ be an arbitrary function defined on graph nodes~$V$, with mean
$\bar{x} = \frac{1}{|V|}\sum_{v \in V} x(v).$
Then
\begin{equation}
	\hat{\bar{x}}\ =\ \frac{\sum_{v \in S} x(v)/d(v)}{\sum_{v \in S} 1/d(v)}
\end{equation}
is an unbiased and consistent estimator of~$\bar{x}$.
By default, we use this reweighting procedure throughput the paper (referring to the combination of random walks with post-hoc reweighting as the RWRW method).

\subsubsection{\label{sec:mg_diagnostics}Multiple Walks and Convergence Diagnostics}

In previous work \cite{Gjoka2010, Gjoka2011_Facebook_JSAC}, we recommended the use of multiple, simultaneous random walks to reduce the chance of obtaining samples that overweight non-representative regions of the graph.  We also recommended the use of formal convergence diagnostics to assess sample quality in an online fashion, which help to determine when a set of walks is in approximate equilibrium, and hence when it is safe to stop sampling.  Use of both multiple walks and convergence diagnostics are critical to effective sampling of OSNs, as our sample case (Section~\ref{sec:mg_lastfm}) illustrates.

In this paper, we use three convergence diagnostics, following \cite{Gjoka2010, Gjoka2011_Facebook_JSAC}.  First, we track the running means for various scalar parameters of interest as a function of the number of iterations.  Second, we use the Geweke \cite{geweke} diagnostic within each random walk, which verifies that mean values for scalar parameters at the beginning of the walk (here the first 10\% of samples) does not differ significantly from the corresponding mean at the end of the walk (here the last 50\%). Third, we use the Gelman-Rubin \cite{gelman-rubin} diagnostic to verify convergence across walks, by ensuring that the parameter variance between walks matches the variance within walks. %

\section{Evaluation in Synthetic Graphs\label{sec:mg_synthetic}}

In this section, we use synthetic graphs to demonstrate two key benefits of the multigraph approach, namely (i) improved connectivity of the union multigraph, even when the underlying individual graphs are disconnected, and (ii) improved mixing time, even when the individual graphs are highly clustered.  The former is necessary for the random walk to converge. The latter determines the speed of convergence.

{\bf Erdos-Renyi graphs.}
In Example~\ref{ex:bernoulli-example} and Fig.~\ref{fig:bernoulli-example}, we noted that even sparse, highly fragmented graphs can have well-connected unions. In Fig.~\ref{fig:ER_graphs}, we generalize this example and quantify the benefit of the multigraph approach. We consider here a collection
$G_{1}, .. ,G_{Q}$
of $Q$ Erdos-Renyi random graphs $(N,p)$ with $N\eq1000$ nodes and $p\eq 1/1000$, i.e., with the expected number of edges $|E|=500$ each. We then look at properties of their multigraph $G$ with increasing numbers of simple graphs $Q$.

\begin{figure}[t!]
\centering
\includegraphics[width=0.49\textwidth]{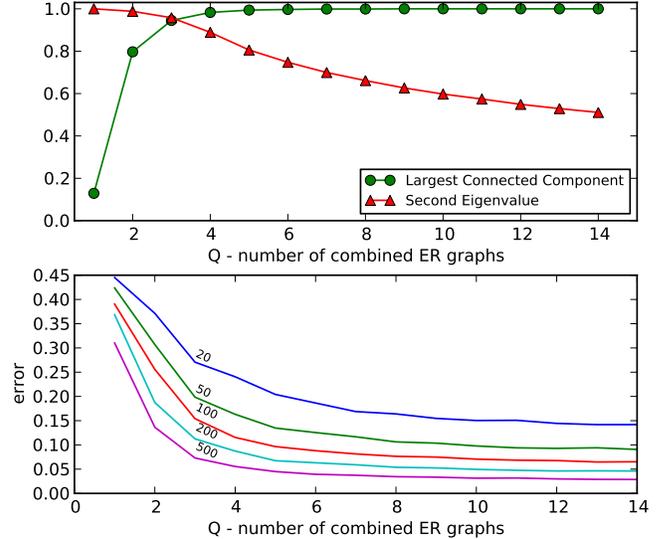}
\caption{Multigraph that combines from several Erdos-Renyi graphs. We generate a collection $G_{1}, .. ,G_{Q}$ of $Q$ random Erdos-Renyi (ER) graphs with $|V|=1000$ nodes and expected $|E|=500$ edges each.
\qquad (top)~We show two properties of multigraph $G$ as a function of $Q$. (1) Largest Connected Component (LCC) fraction ($f_{LCC}$) is the fraction of nodes that belong to the largest connected component in~$G$. (2) The second eigenvalue of the transition matrix of random walk on the LCC is related to the mixing time.
\qquad (bottom)~We also label a fraction $f=0.5$ of nodes within LCC and run random walks of lengths 20\ldots 500 to estimate $f$. We show the estimation error (measured in the standard deviation) as a function of $Q$ (x axis) and walk length (different curves).
}
\label{fig:ER_graphs}
\end{figure}

In order to characterize the connectivity of $G$, we define $f_{LCC}$ as the fraction of nodes that belong to the largest connected component in~$G$. For $Q\eq 1$ we have $f_{LCC}\simeq 0.15$, which means that each simple ER graph is heavily fragmented. Indeed, at least $999$ edges are necessary for connectivity. However, as $Q$ increases, $f_{LCC}$ increases.
With a relatively small number of simple graphs, say for $Q=6$, we get $f_{LCC}\simeq1$, which means that the multigraph is fully connected with high probability. In other words, combining several simple graphs into a multigraph allows us to reach (and sample) many nodes otherwise unreachable.

Note that this example illustrates a more general phenomenon.  Given $Q$ independent random graphs with $N$ nodes each and with expected densities $p_1,\ldots,p_Q$, the probability that an edge $\{u,v\}$ belongs to their union graph $G'$ is $p^*=1-\prod_{i=1}^Q(1-p_i)$.  For $p_i$ approximately equal, this approaches~1 exponentially fast in $Q$.  Asymptotically, the union graph will be almost surely connected where $(N-1) p^*> \ln N$ \cite[pp413--417]{West96}, in which case the union multigraph is also trivially connected.  Thus, intuitively, a relatively small number of sparse graphs are needed for the union to exceed its connectivity threshold.

In order to characterize the mixing time, we plot the second eigenvalue  $\lambda_2$ of the transition matrix of the random walk (on the LCC).  $\lambda_2$ is well-known to relate to the mixing time of the associated Markov chain \cite{Boyd2004_mixing}: the smaller $\lambda_2$, the faster the convergence. In \fig{fig:ER_graphs}(top), we observe that $\lambda_2$ significantly drops with growing~$Q$.
(However, note that adding a new edge to an existing graph does not always guarantee the decrease of $\lambda_2$. It is possible to design examples where $\lambda_2$  increases, although they are rare.)

To further illustrate the connection between $\lambda_2$ and the speed of convergence, we conducted an experiment with a simple practical goal: apply random walk to estimate the size of an exogenously defined ``community'' in the network. We labeled as ``community members'' a fraction $f=0.5$ of nodes within LCC (these nodes were selected with the help of a randomly initiated BFS to better imitate a community).
Next, we ran 100 random walks of lengths 20\ldots 500 within this LCC, and we used them to estimate $f$.
In \fig{fig:ER_graphs} (bottom), we show the standard error of this estimator, as a function of $Q$ (x axis) and walk length (different curves).
This error decreases not only with the walk length, but also with the number $Q$ of combined graphs. This means that by using the multigraph sampling approach we improve the quality of our estimates. Alternatively, we may think of it as a way to decrease the sampling cost. For example, in \fig{fig:ER_graphs}, a random walk of length 500 for $Q\eq 3$ (\ie when LCC is already close to 1) is equivalent to a walk of length 100 for $Q\simeq 8$, which results in a five-fold reduction of the sampling cost.

{\bf ER Graph Plus Random Cliques.}
One may argue that ER graphs are not good models for capturing real-life relations. Indeed, in practice, many relations are highly clustered; \eg a friend of my friend is likely to be my friend. In an extreme case, all members of some community may form a clique. This is quite common in OSNs, where we are often able to browse all members of a group, or all participants of an event.

Interestingly, the multigraph technique is efficient also under the presence of cliques.
 In Fig.~\ref{fig:ER_and_CC}(a), we consider one ER graph, combined with an increasing number of random cliques. 
 We plot the same three metrics as in Fig~\ref{fig:ER_graphs}, and we obtain qualitatively similar results.
This robustness is a benefit of the multigraph approach.

\begin{figure}[ht] 
\centering
\subfigure[Combination of one ER graph ($|V|=200$ nodes and $|E|=100$ edges) with a set of $k-1$  cliques (of size 40 randomly chosen nodes each).]{
\includegraphics[scale=0.5]{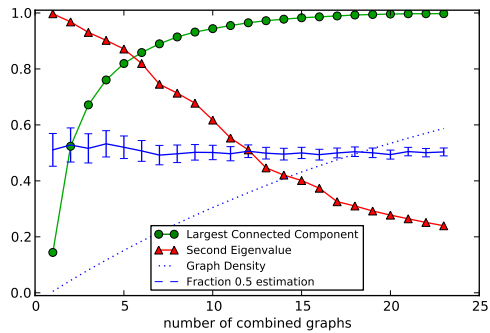}
}
\subfigure[Combination of multiple regular random graphs with clustering~\cite{Newman09_RG_with_clustering}. We set the parameters such that each of $|V|=1000$ nodes has degree equal to~2, and each edge participates in exactly one triangle.]{
\includegraphics[scale=0.46]{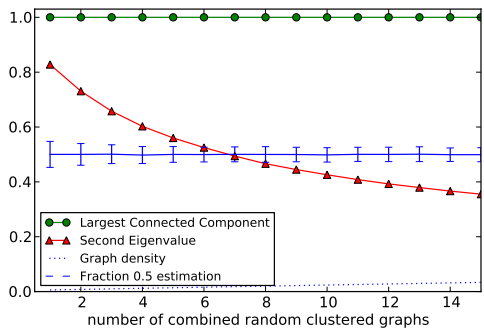}
}
\caption[Optional caption for list of figures]{Multigraphs resulting from a combination of various graphs.}
\label{fig:ER_and_CC}
\end{figure}

{\bf Random Graphs with Clustering.} 
Finally, in Fig~\ref{fig:ER_and_CC}(b), we consider a combination of random graphs with clustering~\cite{Newman09_RG_with_clustering}. 
The results confirm our previous observations. 

\section{Multigraph Sampling of Last.fm\label{sec:mg_lastfm}}

In this section, we apply multigraph sampling to \lastfm - a music-oriented OSN that allows users to create communities of interest that include both listeners and artists.  \lastfm is built around an Internet radio service that compiles a preference profile for each listener and recommends users with similar tastes.  In June 2010, \lastfm was reported to have around 30 million users and was ranked in the top 400 websites in Alexa.
We chose \lastfm to demonstrate our approach because it provides an example of a popular OSN that is fragmented with respect to the social graph (referred to on the site as ``friendship'') as well as other relations.
For example, many \lastfm users mainly listen to music and do not use the social networking features, which makes it difficult to reach them through crawling the friendship graph; likewise, users with similar music tastes may form clusters that are disconnected from other users with very similar music tastes. This intuition was confirmed by our empirical observations.  Despite these challenges, we show that multigraph sampling is able to obtain a fairly representative sample in this case, while single graph sampling on any specific relation fails.

\subsection{Crawling \lastfm}

We sample \lastfm via random walks on several individual relations as well as on their union multigraph.
Fig~\ref{fig:collected_userinfo} shows the information collected for each sampled user.

\begin{figure}[t]
\centering
\includegraphics[scale=0.40, angle=0]{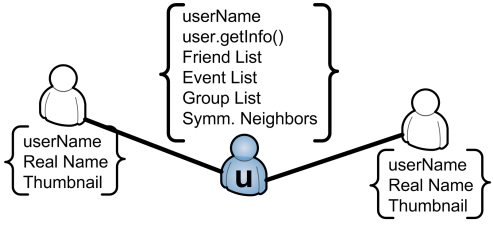}
\caption{Information collected for a sampled user $u$.
 (a) userName and user.getInfo: Each user is uniquely identified by her userName. The API call user.getInfo returns : real Name, userID, country, age, gender, subscriber, playcount, number of playlists, bootstrap, thumbnail, and user registration time.
(b) Friends list: List of mutually declared friendships.
(c) Event list. List of past and future events that the user indicates she will attend. We store the eventID and number of attendees.
(d) Group list. List of groups of which the user is a member. We store the group name and group size.
(e) Symmetric neighbors. List of mutual neighbors.
}
\label{fig:collected_userinfo}
\end{figure}

\subsubsection{Walking on Relations}

We consider the following relations between two users:
\begin{itemize}
\item { \textit{Friends:}} This refers to mutually declared friendship between two users.
\item { \textit{Groups:}} Users with something in common are allowed to start a group. Membership in the same group connects all involved users.
\item { \textit{Events:}} \lastfm allows users to post information on concerts or festivals. Attendees can declare their intention to participate. Attendance in the same event connects all involved users.
\item { \textit{Neighbors:}} \lastfm matches each user with up to 50 similar neighbors based on common activity, membership, and taste. The details of neighbor selection are proprietary.  We symmetrize this directed relation by considering only mutual neighbors as adjacent.
\end{itemize}

First, we collect a sample of users by a random walk on the graph for each individual relation, that is Friends, Groups, Events, and Neighbors.  Then, we consider sets of relations (namely: Friends-Events, Friends-Events-Groups, and Friends-Events-Groups-Neighbors) and we perform a random walk on the corresponding union multigraph. In the rest of the section, we refer to random walks on different simple graphs or multigraphs as {\em  crawl types}.

\begin{table*}[t!]
\footnotesize
\begin{center}
\begin{tabular}{|@{}l|@{}l|@{}l|@{}l|@{}l|@{}l|@{}l|@{}l|@{}l|}
\hline
Crawltype               & Friends      &  Events      & Groups       &  Neighbors    & Friends-Events & Friends-Events- &  Friends-Events-   &       UNI        \\
                        &              &              &              &               &                &     Groups      &  Groups-Neighbors  &               \\
\hline
\# Total Users          & 5$\times50K$ & 5$\times50K$ & 5$\times50K$ &  5$\times50K$ & 5$\times50K$  &  5$\times50K$   &    5$\times50K$    &         500K     \\
\% Unique users         &   71.0\%     &   58.5\%     &   74.3\%     & 53.1\%        &    59.4\%     &  75.5\%         &      75.6\%       &     99.1\%             \\
\# Users kept           &   245K       &   245K       &   245K        & 245K         &    200K       &  187K           &      200K         &     500K             \\
Crawling period         &  07/13-07/16 &  07/13-07/18 & 07/13-07/17  & 07/13-07/17   & 07/13-07/18   &  07/13-07/18    &     07/13-07/21   &    07/13-07/16   \\
\hline
Avg \# friends          &   10.7       &   18.0       &   15.8       &  12.2         &     9.8       &  6.8            &      6.6          &      1.2        \\
Avg \# groups           &   2.40       &    4.71      &   5.22        &   2.90         &    2.47       &  0.71           &     0.67          &     0.30             \\
Avg \# events           &  2.44/0.17   &    7.49/0.56 &   3.96/0.28  &  2.94/0.27    &   2.30/0.17   &  0.74/0.05      &    0.73/0.04      &   0.28/0.02                \\
 (past/future)          &              &              &              &               &               &                 &                   &                  \\
\hline
\end{tabular}
\end{center}
\vspace{-3pt}
\caption{Summary of collected datasets in July 2010. The percentage of users kept is determined from convergence diagnostics. Averages shown are after convergence and re-weighting which corrects sampling bias.}
\label{tab:datasets}
\vspace{-17pt}
\end{table*}

\subsubsection{Uniform Sample of userIDs (UNI)\label{sec:mg_uni}}

\lastfm usernames uniquely identify users in the API and HTML interface. However, internally, \lastfm associates each username with a userID,  presumably used to store  user information in the internal database.  We discovered that it is possible to obtain usernames from their userIDs, a fact that allowed us to obtain a uniform, ``ground truth'' sample of the user population.  Examination of registration and ID information indicates that \lastfm allocates userIDs in an increasing order. Fig. \ref{fig:userid_space} shows the exact registration date and the assigned userID for each sampled user in our crawls obtained through exploration.  With the exception of the first $\sim2M$ users (registered in the first 2 years of the service), for every $userID_1>userID_2$  we have $registration\_time(userID_1) > registration\_time(userID_2)$. We also believe that userIDs are assigned sequentially because we rarely observe non-existent userIDs  after the $\sim2,000,000$ threshold. We conjecture that the few non-existent userIDs after this threshold are closed or banned accounts. At the beginning of the crawl, we found no indication of user accounts  with IDs above $\sim31,200,000$. Just before the crawls, we registered new users that were assigned user IDs slightly higher than the latter value.

Using the userID mechanism, we obtained a reference sample of \lastfm users by uniform rejection sampling \cite{leon-garcia}.  Specifically, each user was sampled by repeatedly drawing uniform integers between $0$ and $35$ million (\ie the maximum observed ID plus a $\sim$4 million ``safety'' range) and querying the userID space.  Integers not corresponding to a valid userID were discarded, with the process being repeated until a match was obtained. IDs obtained in this way are uniformly sampled from the space of user accounts, irrespective of how IDs are actually allocated within the address space \cite{Gjoka2011_multigraph_JSAC}.  We employ this procedure to obtain a sample of 500K users, referred to here as ``UNI.''
We note that the same method has been recently used in~\cite{Gjoka2010}, as well as in~\cite{Rejaie2010}. The latter also examined population growth and active vs. inactive users, which are out of the scope of this paper.

\begin{figure}[t]
\centering
\includegraphics[scale=0.25, angle=0]{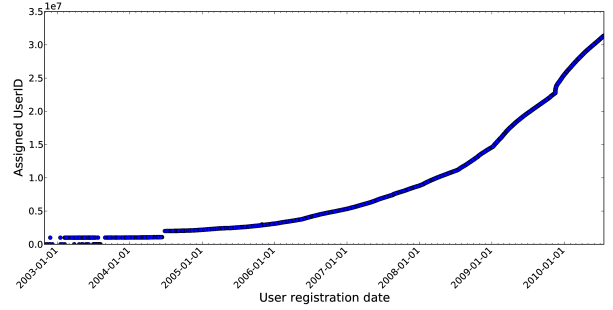}
\caption{\lastfm assigns userIDs in increasing order after 2005: userID vs registration time. }
\label{fig:userid_space}
\end{figure}

Although UNI sampling currently solves the problem of uniform node sampling in \lastfm and is a valuable asset for this study, it is not a general solution for sampling OSNs. Such an operation is not generally supported by OSNs.  Furthermore, the userID space must not be sparse for this operation to be efficient. In the \lastfm case, the small userID space makes this possible at the time of this writing; however, a simple increase of the userID space to 48 or 64 bits would render the technique infeasible.  In summary, we were able to obtain a uniform sampling of userIDs and use it as a baseline for evaluating the sampling methods of interest against the target distribution.  

\subsubsection{\label{sec:mg_population}Estimating Last.fm population size}
In addition to the UNI sample  presented in Table \ref{tab:datasets}, we  obtained a second UNI sample of the same size one week later. We then applied the capture-recapture method \cite{seber1982estimation} to estimate the Last.fm user population during the period of our crawling. According to this method, the population size is  estimated to be :
$$ P_{\lastfm} = \frac{N_{UNI1}\times N_{UNI2}}{R} = 28.5M,$$
where $N_{UNI1} = N_{UNI2} = 500K$ and $R$ is the number of valid common userIDs sampled during the first and second UNI samples. This estimation is consistent with our observations of the maximum userID space and close to the reported size of \lastfm on various Internet websites.  We will later use this second sample to comment on the topology change during our crawls.

\subsubsection{\label{sec:mg_timescale}Topology Change During Sampling Process}
While substantial change could in theory affect the estimation process, \lastfm evolves very little during the duration of our crawls. The increasing-order userID assignment allows us to infer the maximum growth of \lastfm during this period, which we estimate at $25K/$day on average.  Therefore, with a population increase of $0.09\%/$day, the user growth during our crawls (2-7 days) is calculated to range between $0.18\%-0.63\%$ per crawl type, which is quite small. Furthermore, the comparison between the two UNI samples revealed almost identical distributions for the properties studied here, as shown in Fig~\ref{fig:unicompare_mg}.
Therefore, in the rest of the paper, we assume that any changes in the \lastfm network during the crawling period can be ignored. This is unlike the context of dynamic graphs, where considering the dynamics is essential, \eg see \cite{Rasti2008, Willinger09-OSN_Research, acer2010random}.

\begin{figure*}[t]
\centering
\subfigure[Number of friends]{
\includegraphics[scale=0.38, angle=0]{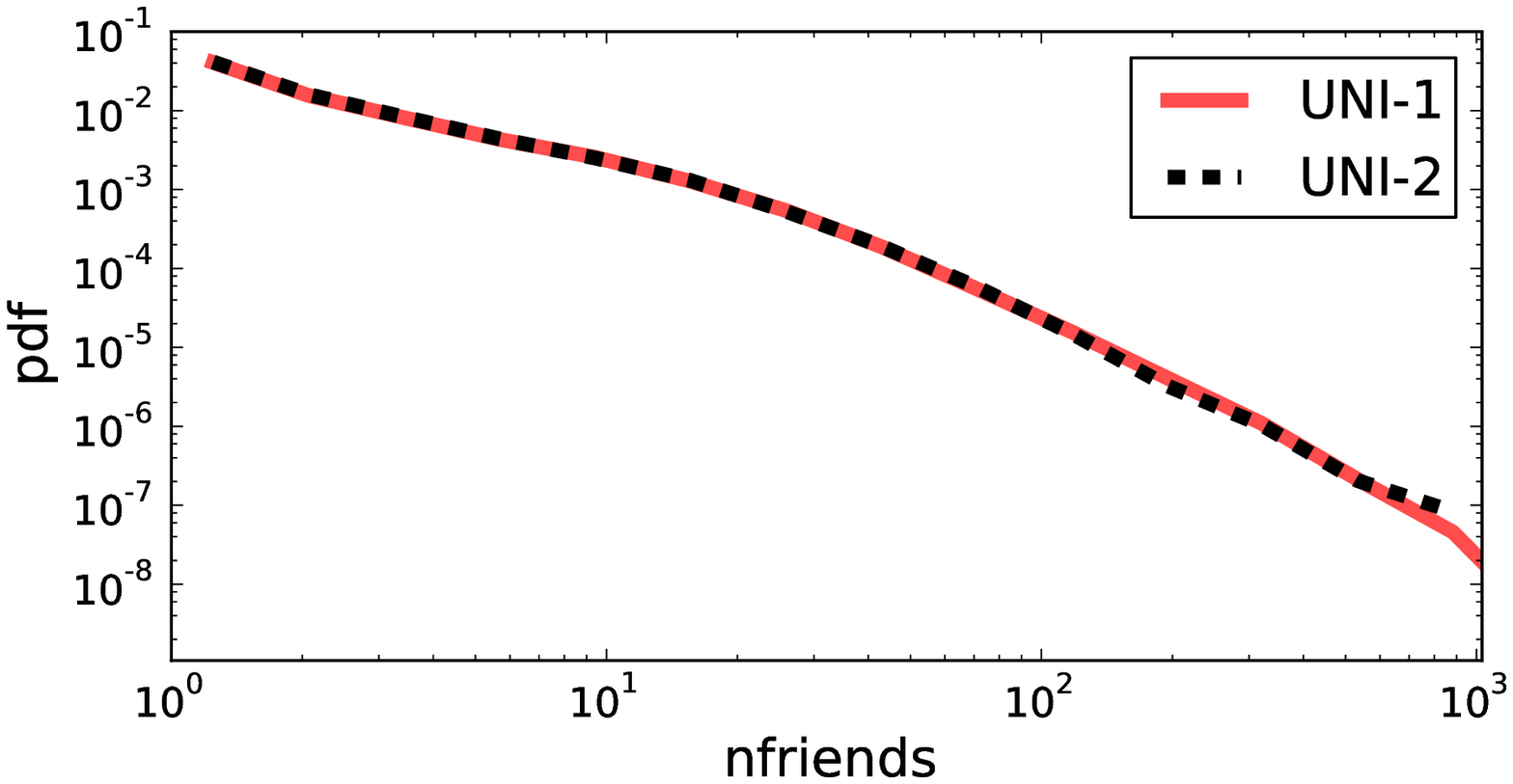}}
\subfigure[Number of groups]{
\includegraphics[scale=0.38, angle=0]{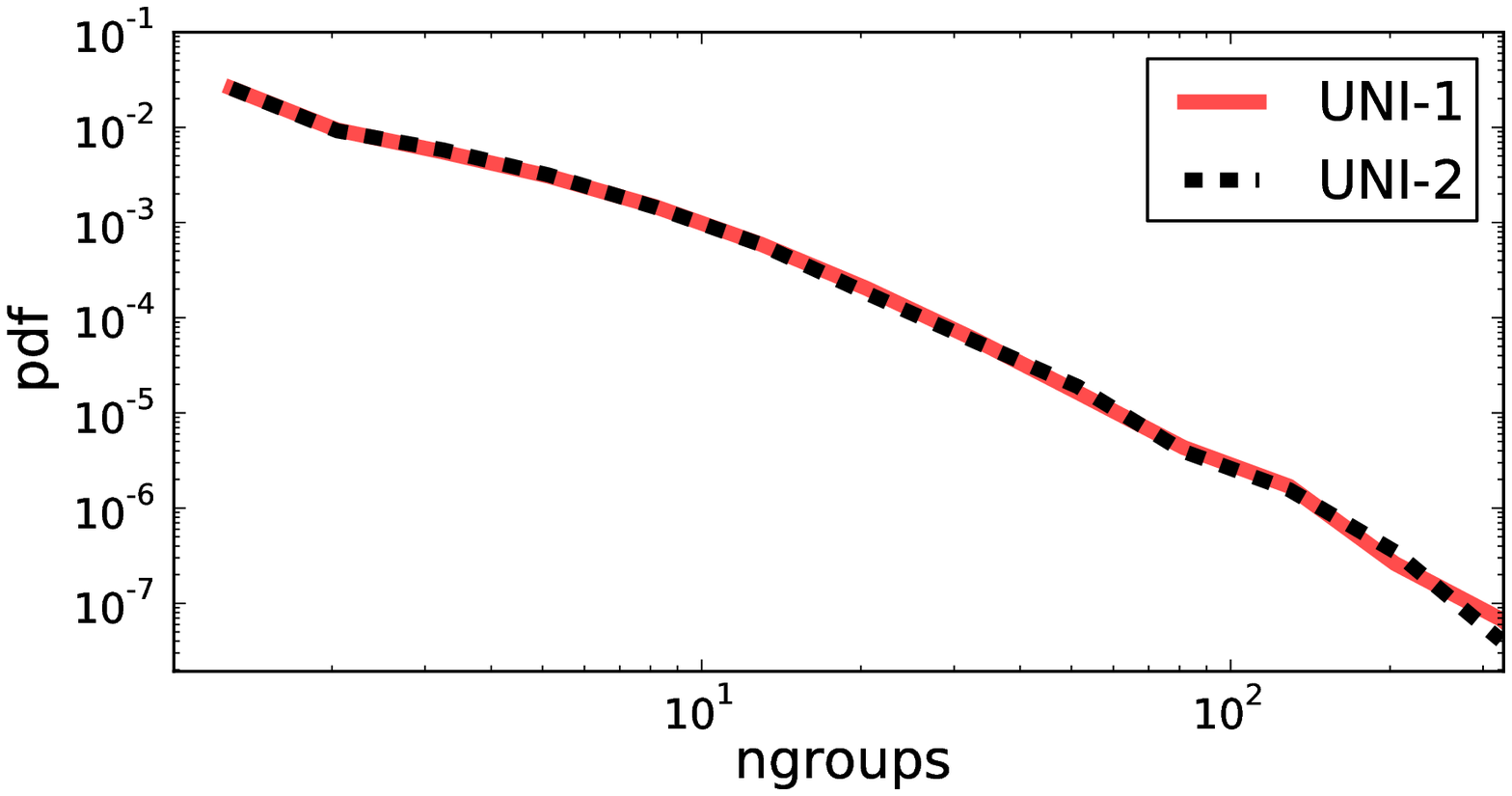}}
\subfigure[Number of past events ]{
\includegraphics[scale=0.38, angle=0]{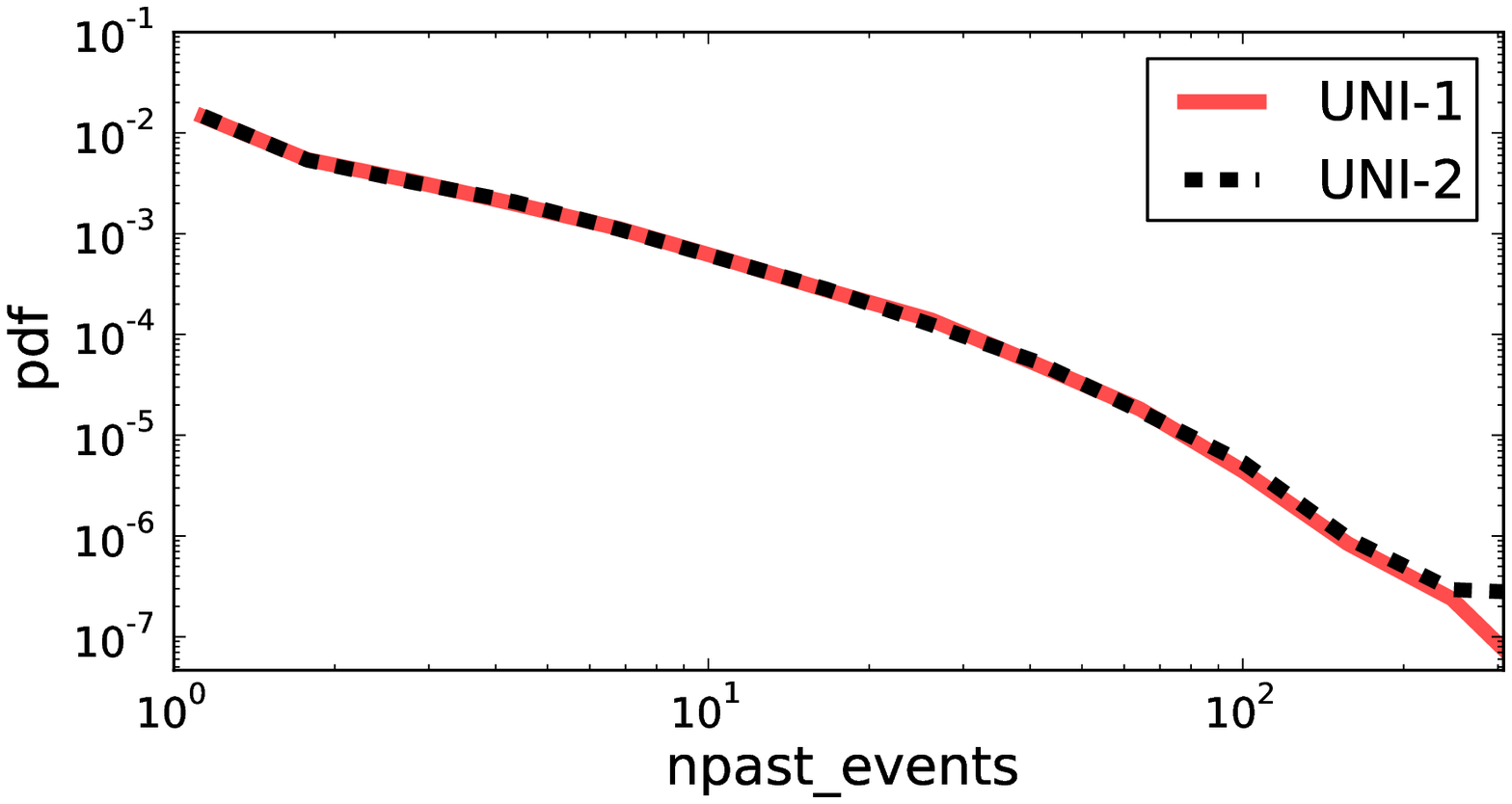}}
\subfigure[Number of future events ]{
\includegraphics[scale=0.38, angle=0]{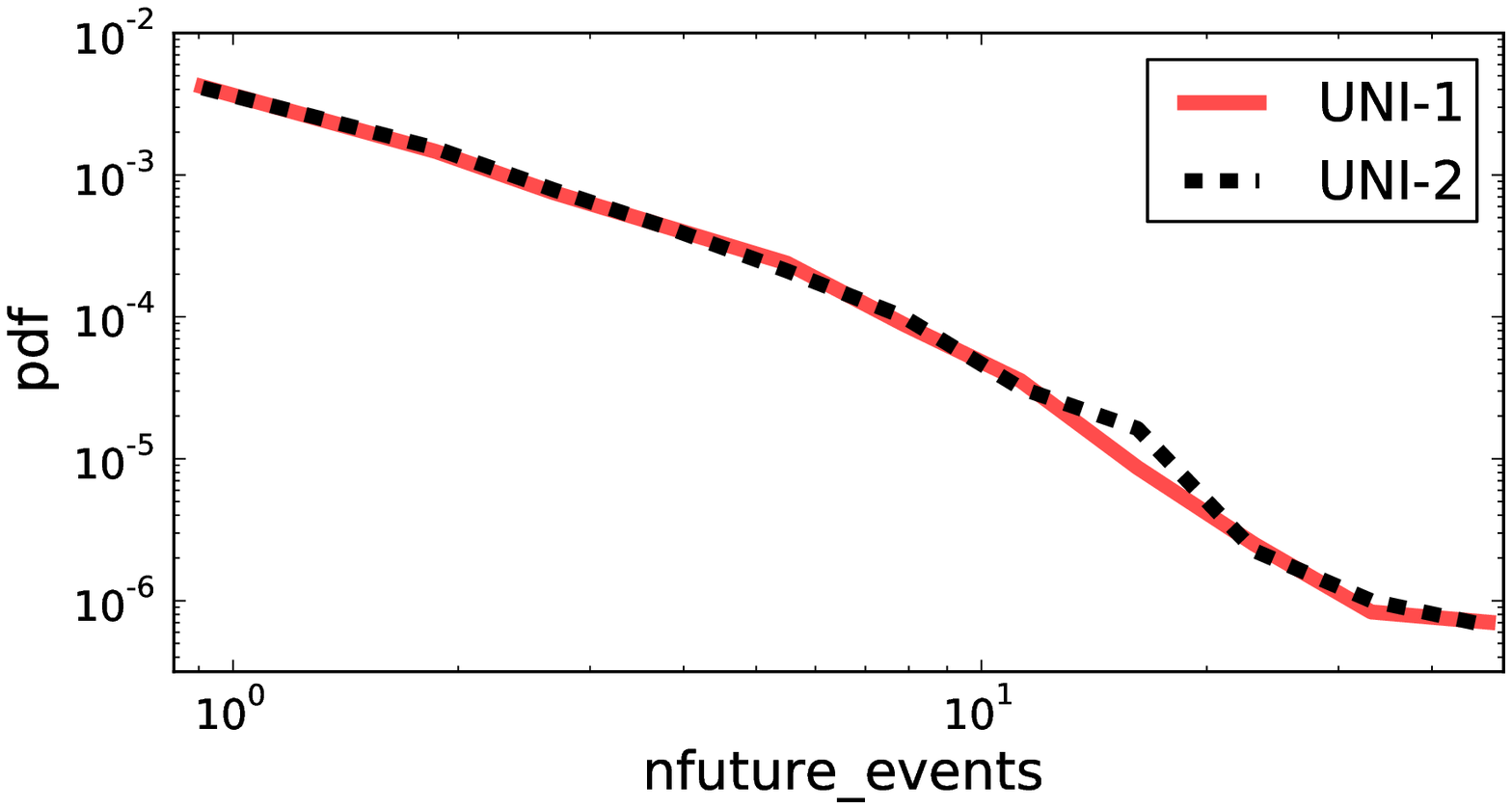}}
\caption{Probability distribution function (PDF) of UNI samples obtained one week apart from each other. Results are binned.}
\label{fig:unicompare_mg}
\end{figure*}

\subsubsection{Efficient Multigraph Sampling in \lastfm${}$}

To collect data from \lastfm, we use a combination of API calls and data scraping. Consider that we are sampling user $u$. For efficient implementation of multigraph sampling we proceed in two stages, as shown in Fig~\ref{fig:multigraph}(f). 

In the first stage, we discover the graphs of user $u$, and $u$'s degrees in them.
In our study, we use the API calls {\tt user.getfriends}, {\tt user.getneighbors}, {\tt user.getpastevents}, and {\tt user.getevents} to collect the list of friends, neighbors, past events, and future events respectively. Due to a lack of an API call that lists the groups of a user, we use data scraping to collect the list of groups and corresponding size for each group. We treat each individual group and event as a different graph in the multigraph. We also consider the set of friends and neighbors to comprise the friends and neighbors graph respectively in the multigraph.  At the end of the first stage, we select one of the graphs $G_i$ in accordance with Algorithm~\ref{alg:mixture} in Section~\ref{sec:mg_samplingmethodology}.

We should note that at the end of the first stage, we have not enumerated any user from any of the groups and events graphs. Each of these graphs is quite large (up to tens of thousands of users) and depending on the user, there are many groups or events per user (up to thousands).  On the other hand, we have enumerated users of friends and neighbors since knowledge of neighborhood size is equivalent to enumeration  \footnote{There might be workarounds to enumerating friends but they are not necessarily more efficient. For example, we could extract the number of friends by data scraping.  In general, in another setting we could do away with any kind of enumeration in the first stage.} for these graphs.  Overall, our two stage approach saves us bandwidth and time by avoiding the enumeration of users for graphs that we are not going to sample from at each iteration.

In the second stage, we pick uniformly at random one of the nodes from the graph $G_i$, selected at the end of the first stage. If the graph $G_i$ is a graph of a group or an event, we need to carefully implement this action to be efficient. More specifically, we do not need to enumerate all group members or event attendants from a group or event graph. Instead, we can take advantage of the pages functionality that OSNs often provide and only fetch the page that corresponds to the user selected uniformly at random. In our study, to fetch group members we use the API call {\tt group.getmembers}, which returns 50 users per page. To fetch event attendants we use data scraping\footnote{We prefer data scraping to the API call {\tt event.getattendees} because the API call i) is not paged ii)  does not return users that marked  ``maybe`` for the event iii) is very slow for large events.}, which also returns  50 users per HTML page.

\subsubsection{Data Collection}
We used a cluster of machines to execute all crawl types under comparison simultaneously.  For each crawl type, we run $|V_0|=5$ different independent walks. The starting points for the five walks, in each crawl type, are randomly selected users identified by the web site as listeners of songs from each of five different music genres: country, hip hop, jazz, pop and rock.  This set was chosen to provide an overdispersed seed set, while not relying on any special-purpose methods (\eg UNI sampling).  To ensure that differences in outcomes do not result from choice of seeds, the same seed users are used for all crawl types.  We let each independent crawl continue until we determine convergence per walk and per crawl, using online diagnostics as introduced in \cite{Gjoka2010} and described in Section \ref{sec:mg_diagnostics}. Eventually, we collected  exactly 50K samples for each random walk crawl type.
 Finally, we collect a UNI sample of 500K users.

\begin{table}[t!]
\begin{center}
\begin{tabular}{@{}l@{}||@{\,}l|@{\,}l|@{\,}l|@{\,}l}
Crawl Type                       & Friends  & Events  & Groups  & Neighbors  \\
                                 & Graph   & Graphs   & Graphs  &  Graphs    \\
\hline \hline
Friends                          & 100\%   & 0\%      & 0\%     &    0\%     \\
\hline
Events                           &   0\%   &  100\%   & 0\%     &    0\%     \\
\hline
Groups                           &   0\%   &   0\%    & 100\%   &    0\%     \\
\hline
Neighbors                        &   0\%   &   0\%    &   0\%   &  100\%     \\
\hline
Friends-Events                   & 2.2\%   &   97.8\% & 0\%     &    0\%     \\
\hline
Friends-Events-Groups            & 0.3\% &   5.4\%  & 94.3\%    &    0\%     \\
\hline
Friends-Events-Groups-Neighbors  & 0.3\% &   5.5\%  & 94.2\%    &   0.02\%   \\
\end{tabular}
\end{center}
\caption{Percentage of time a particular graph (edges corresponding to this graph) is used during the crawl by Algorithm~\ref{alg:mixture}}
\label{tab:percentage_multigraph}
\end{table}

\subsubsection{Summary of Collected Datasets}

Table \ref{tab:datasets} summarizes the collected datasets. Each crawl type contains $5 \times 50K=250K$ users.  We observe that there is a large number of repetitions in the random walks of each crawl type, ranging from 25\% (in Friends-Events-Groups-Neighbors) to 47\% (in Neighbors).  This appears to stem from the high levels of clustering observed in the individual networks. It is also interesting to note that the crawling on the multigraph Friends-Events-Groups-Neighbors is able to reach more unique nodes than any of the single graph crawls.

Table \ref{tab:percentage_multigraph} shows the fraction of Markov chain transitions using each individual relation. The results for the single-graph crawl types Friends, Events, Groups, and Neighbors are as expected: they use their own edges 100\% of the time and other relations' 0\%. Besides that, we see that Events relations  dominate Friends when they are combined in a multigraph, and Groups dominate Friends, Events, and Neighbors when combined with them.  This occurs because many groups and events are quite large (hundreds or thousands of users), leading participants to have very high relationship-specific degree for purposes of Algorithm~\ref{alg:mixture}, and thus for the Group or Event relations to be chosen more frequently than low-degree relations like Friends. In the crawl types obtained through a random walk, the highest overlap of users is observed between Groups and Friends-Events-Groups-Neighbors (66K) while the lowest is between Neighbors and Friends-Events (5K). It is noteworthy that despite the dominance of Groups and the high overlap between Groups and Friends-Events-Groups-Neighbors, the aggregates for these two crawl types in Table \ref{tab:datasets} lead to very different samples of users.

\begin{figure}[t]
\centering
\includegraphics[scale=0.35, angle=270]{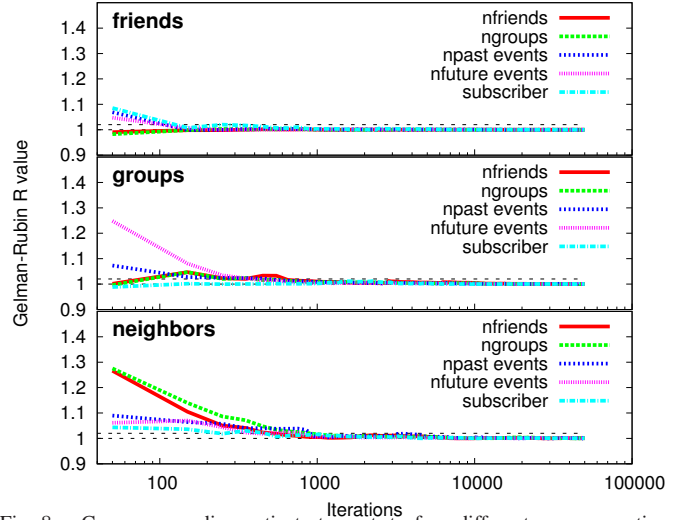}
\caption{Convergence diagnostic tests w.r.t. to four different user properties (``nfriends'': number of friends, ``ngroups'': number of groups, ``npast\_events'': number of past events, ``nfuture\_events'': number of future events, and ``subscriber'') and three different crawl types (Friends, Groups, Neighbors).}
\label{fig:gelman_rubin_single}
\end{figure}

\subsection{Evaluation Results}

\subsubsection{Convergence}

{\em Burn-in.} To determine the burn-in for each crawl type in Table~\ref{tab:datasets}, we run the Geweke diagnostic separately on each of its 5 chains, and the Gelman-Rubin diagnostic across all 5 chains at once, for several different properties of interest. The Geweke diagnostic shows that first-order convergence is achieved within each walk after approximately 500 iterations at maximum. For the single relation crawl types (Friends, Events, Groups, and Neighbors), the Gelman-Rubin  diagnostic indicates that convergence is attained within 1000 iterations per walk  (target value for R below 1.02 and close to 1) as shown in Fig.~\ref{fig:gelman_rubin_single}.

On the other hand,  multigraph crawl types take longer to reach equilibrium. Fig.~\ref{fig:gelman_rubin_multi} presents the Gelman-Rubin (R) score for three multigraph crawl types (namely Friends-Events, Friends-Events-Groups, Friends-Events-Groups-Neighbors) and five user properties (namely number of friends, number of groups, number of past/future events, and subscriber - a binary value which indicates whether the user is a paid customer).  We observe that it takes 10K, 12.5K and 10K samples for each crawl type correspondingly to converge. However, as we show next, they include more isolated users and better reflect the ground truth, while the single graph sampling methods fail to do so.  This underscores an important point regarding convergence diagnostics: while useful for determining whether a random walk sample approximates its equilibrium distribution, they cannot reliably identify cases in which the equilibrium itself is biased (\eg due to non-connectivity). For the rest of the analysis, we discard the number of samples each crawl type needed to converge.

\begin{figure}[t]
\centering
\includegraphics[scale=0.35, angle=270]{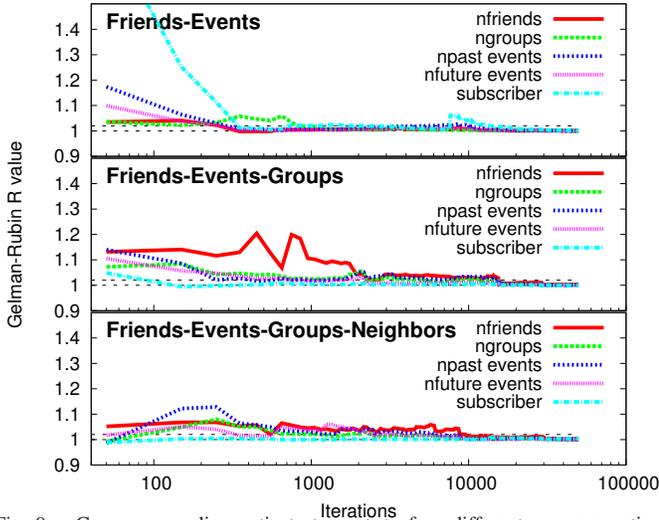}
\caption{Convergence diagnostic tests w.r.t. to four different user properties (``nfriends'': number of friends, ``ngroups'': number of groups, ``npast\_events'': number of past events, ``nfuture\_events'': number of future events, and ``subscriber'') and three different crawl types (Friends-Events, Friends-Events-Groups, Friends-Events-Groups-Neighbors).}
\label{fig:gelman_rubin_multi}
\end{figure}

{\em Total Running Time.} Before we analyze the collected datasets, we verify that the remaining walk samples, after discarding burn-in, have reached their stationary distribution. Table \ref{tab:datasets} contains the ``Number of users kept'' for each crawl type. We use the convergence diagnostics on the remaining samples to assess convergence formally. The results are qualitatively similar to the burn-in determination section. We also perform visual inspection of the running means  in Fig~\ref{fig:running_means} for four different properties, which reveals that the  estimation of the average for each property  stabilizes within 2-4k samples per walk (or 10k-20k over all 5 walks).

\subsubsection{Discovering Isolated Components}

As noted above, part of our motivation for sampling \lastfm using multigraph methods stems from its status as a fragmented network with a rich multigraph structure.  In particular, we expected that large parts of the user base would not be reachable from the largest connected component in any one graph.  Such users could consist of either isolated individuals or highly clustered sets of users lacking ties to rest of the network. We here call {\em isolate}, any user that has degree 0 in a particular graph relation. Walk-based sampling on that particular graph relation has no way of reaching those isolates, but a combination of graphs might be able to reach them, assuming that a typical user participates in different ways in the network (\eg a user with no friends may still belong to a group or attend an event).

\begin{table}[t!]
\begin{center}
\begin{tabular}{@{}l@{\,}||@{\,}l@{\,}|@{\,}l@{\,}|@{\,}l@{\,}|@{\,}l@{\,}}
Crawl Type                       &  Friends   &  Future       &  Past       & Groups \\
                                 &  Isolates  &   Events      &  Events     & Isolates\\ 
                                 &            &  Isolates     & Isolates     & \\
\hline \hline
Friends                          &   0\%      &   93.7\%      &    73.2\%   & 60.4\%  \\
\hline
Events                           & 19.2\%     &   78.2\%      &    4.5\%    &  41.7\%  \\
\hline
Groups                           & 21.2\%     &   89.9\%      &    62.0\%   & 0.0\% \\
\hline
Neighbors                        & 40.4\%     &   89.5\%      &    71.2\%   & 62.4\% \\
\hline
Friends-Events                   & 6.2\%     &   93.5\%      &    69.9\%    & 61.6\% \\
\hline
Friends-Events-Groups            & 5.5\%      &   98.15\%      &    88.1\%   & 85.3\% \\
\hline
Friends-Events-Groups-Neighbors  & 7.4\%      &   98.3\%      &    86.7\%   & 86.3\% \\
\hline
UNI                              & 87.9\%    &   99.2\%      &   96.1\%    &  93.8\% \\
\end{tabular}
\end{center}
\caption{Percentage of sampled nodes that are isolates (have degree 0) w.r.t. to a particular (multi)graph.}
\label{tab:percentage_isolates}
\end{table}

In Table \ref{tab:percentage_isolates}, we report the percentage of nodes in each crawl type that are estimated to be isolates, and compare this percentage to the UNI sample. Observe that there is an extremely high percentage of isolate users in any single graph: \eg  UNI samples are ~88\% isolates in the Friends relation,  96-99\% isolates in the Events relation, and 93.8\% isolates in the Groups relation. Such isolates are not necessarily inactive: for instance, 59\% of users without friends have either a positive playcount or playlist value, which means that they have played music (or recorded their offline playlists) in \lastfm, and hence are or have been active users of the site. This confirms our expectation that \lastfm is indeed a fragmented graph.

More importantly, Table \ref{tab:percentage_isolates} allows us to assess how well different crawl types estimate the \% of users that are isolates with respect to a particular relation or set of relations.  We observe that the multigraph that includes all relations (Friends-Events-Groups-Neighbors) leads to the best estimate of the ground truth (UNI sample - shown in the last row). The only exception is the friends isolates, where the single graph Neighbors gives a better estimate of the percentage of isolates over all other crawl types. The multigraph crawl type Friends-Events-Groups-Neighbors uses the Neighbors relation only 0.02\% of of the time, and thus does not benefit as much as might be expected (though see below).
A weighted random walk that put more emphasis on this relation (or use of a relation that is less sparse) could potentially improve performance in this respect. 

\begin{figure*}[t]
\centering
\subfigure[Number of friends]{\label{fig:mean_nfriends}
\includegraphics[scale=0.23, angle=0]{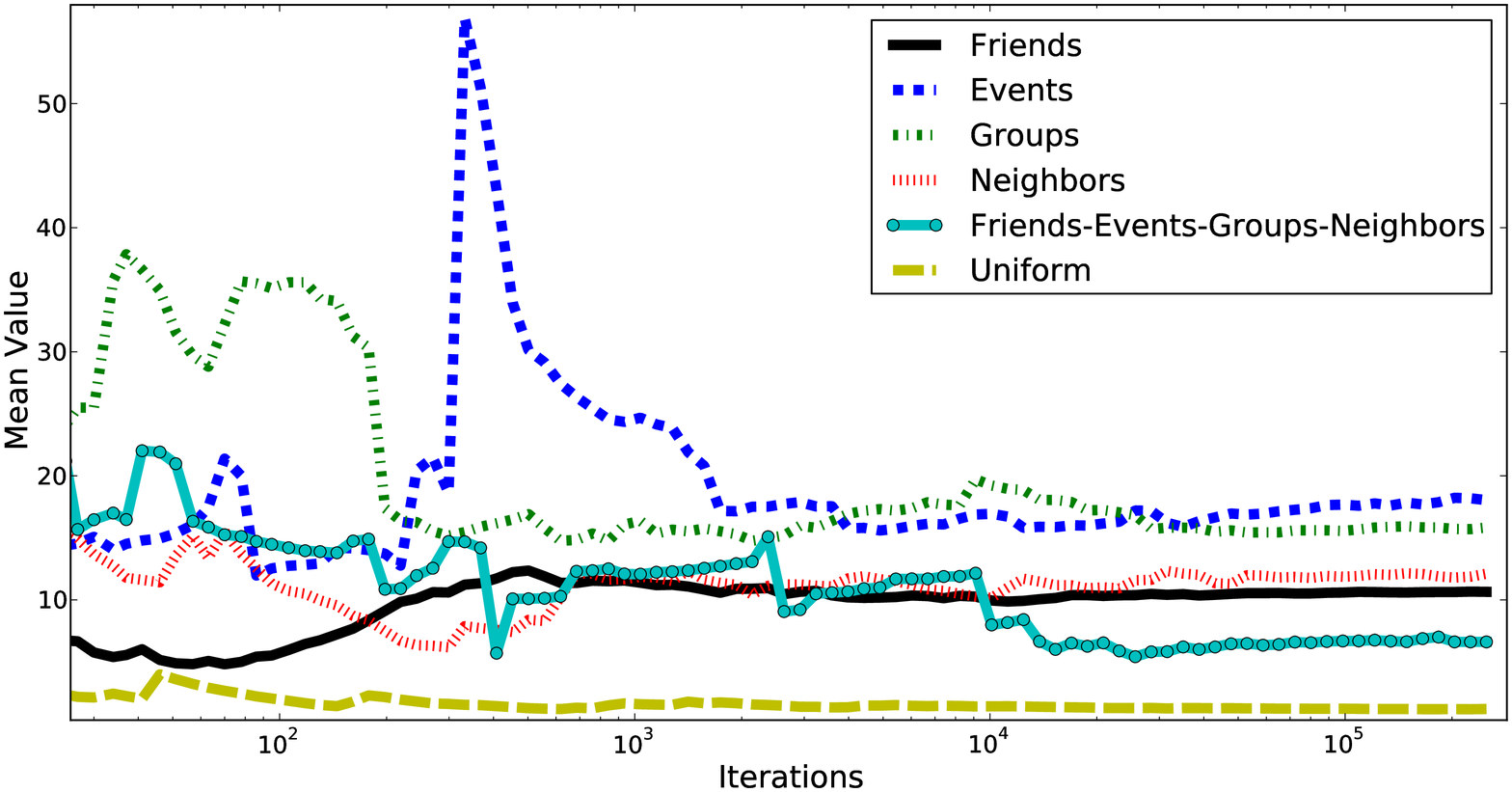}}
\subfigure[Number of groups]{\label{fig:mean_ngroups}
\includegraphics[scale=0.23, angle=0]{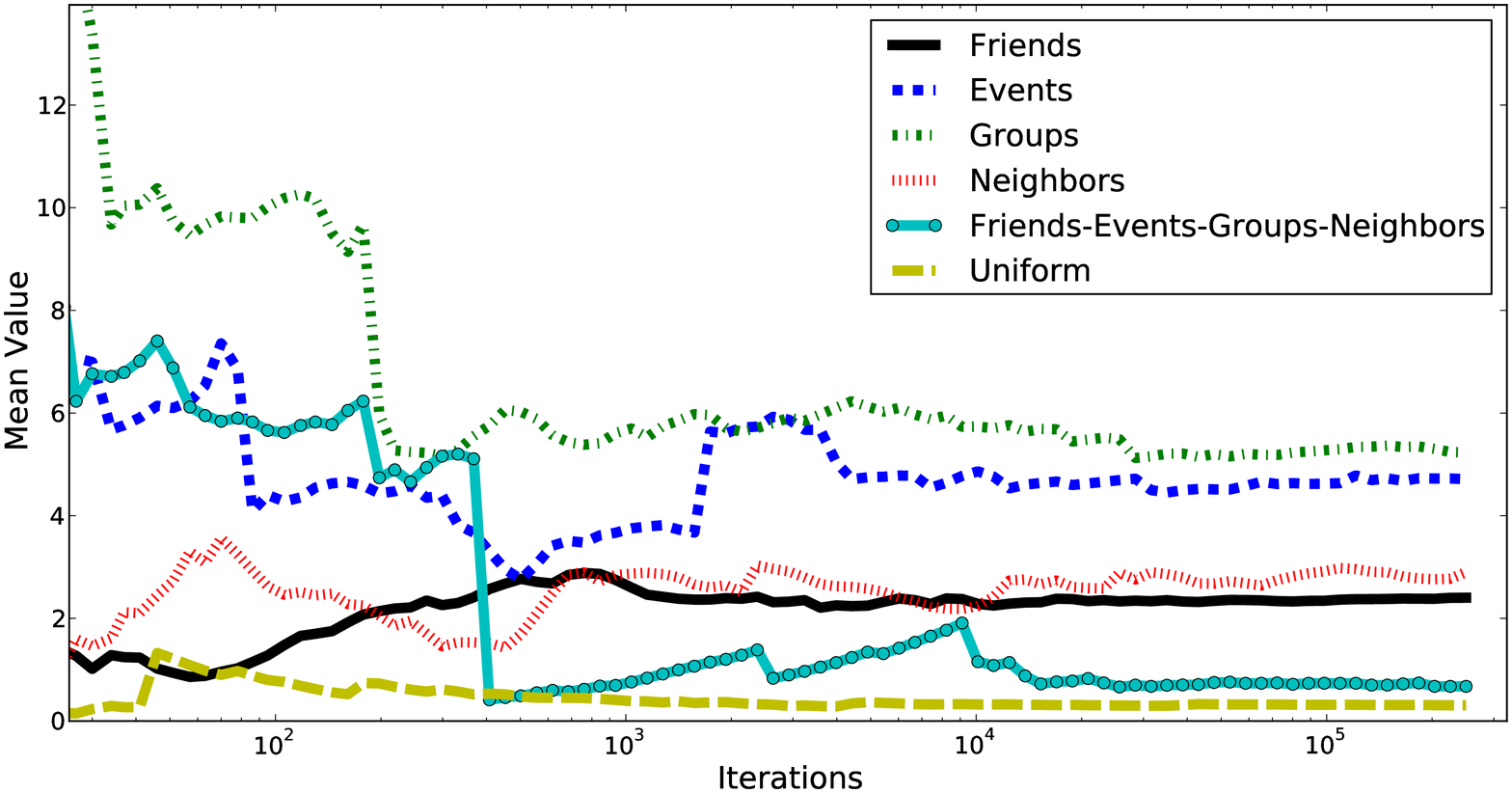}}
\subfigure[Number of past events]{\label{fig:mean_past_events}
\includegraphics[scale=0.23, angle=0]{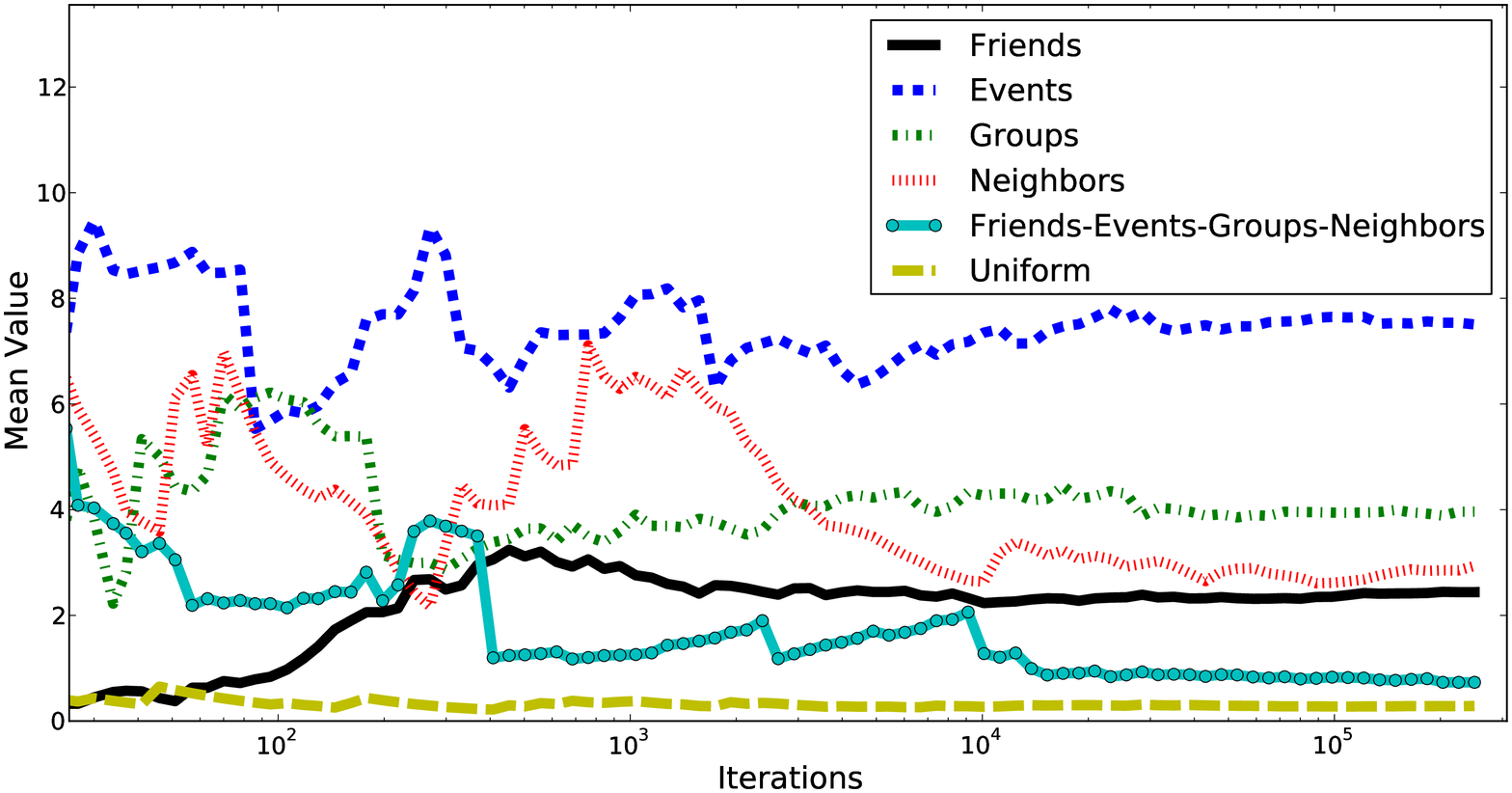}}
\subfigure[\% of Subscribers]{\label{fig:mean_subscriber}
\includegraphics[scale=0.23, angle=0]{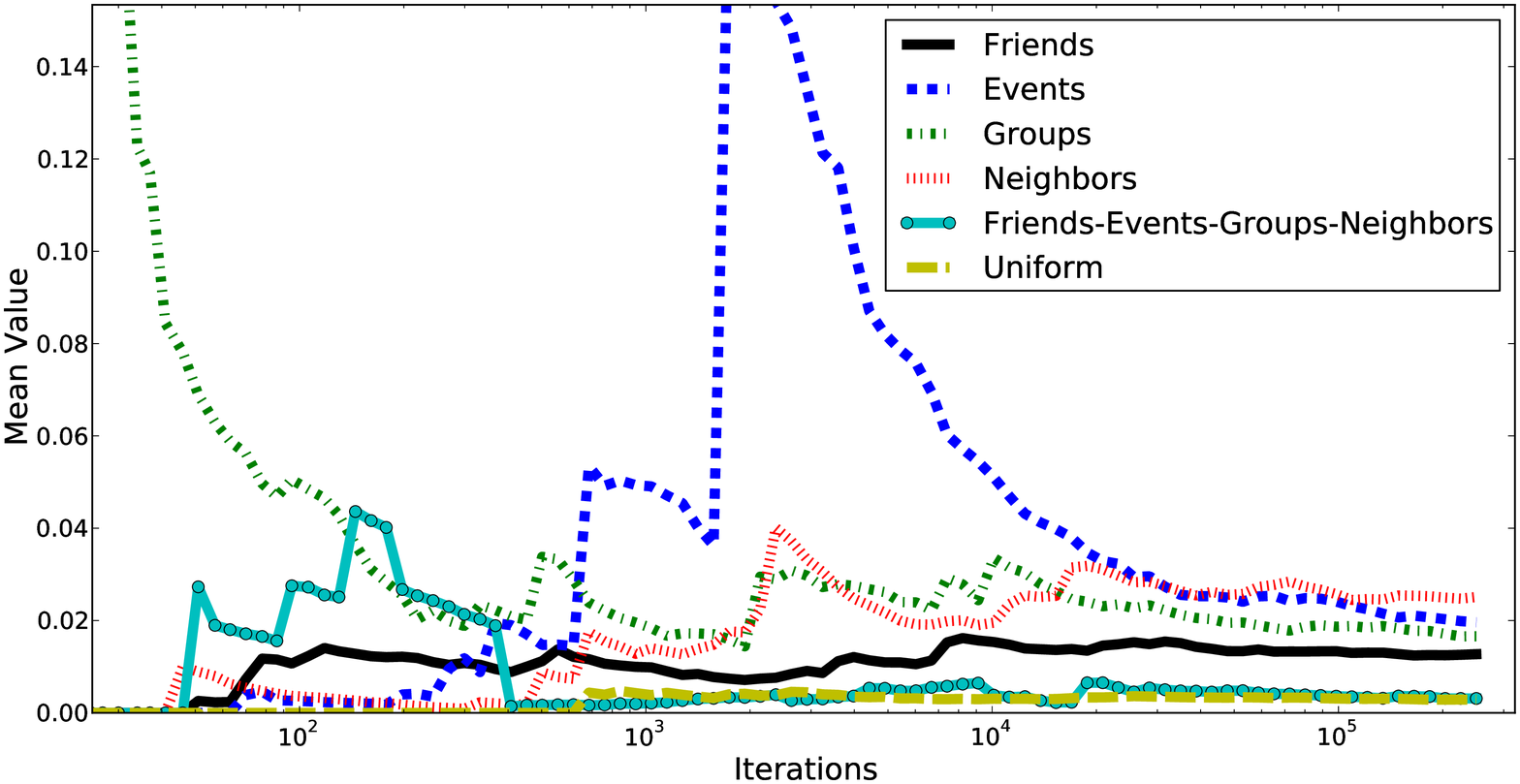}}
\caption{Single graph vs multigraph sampling. Sample mean over the number of iterations, for four user properties (number of friends, number of groups, number of past events, \% subscribers), as estimated by different crawl types.}
\label{fig:running_means}
\end{figure*}

\begin{figure*}[t]
\centering
\subfigure[Number of friends]{\label{fig:pdf_nfriends}
\includegraphics[scale=0.23, angle=0]{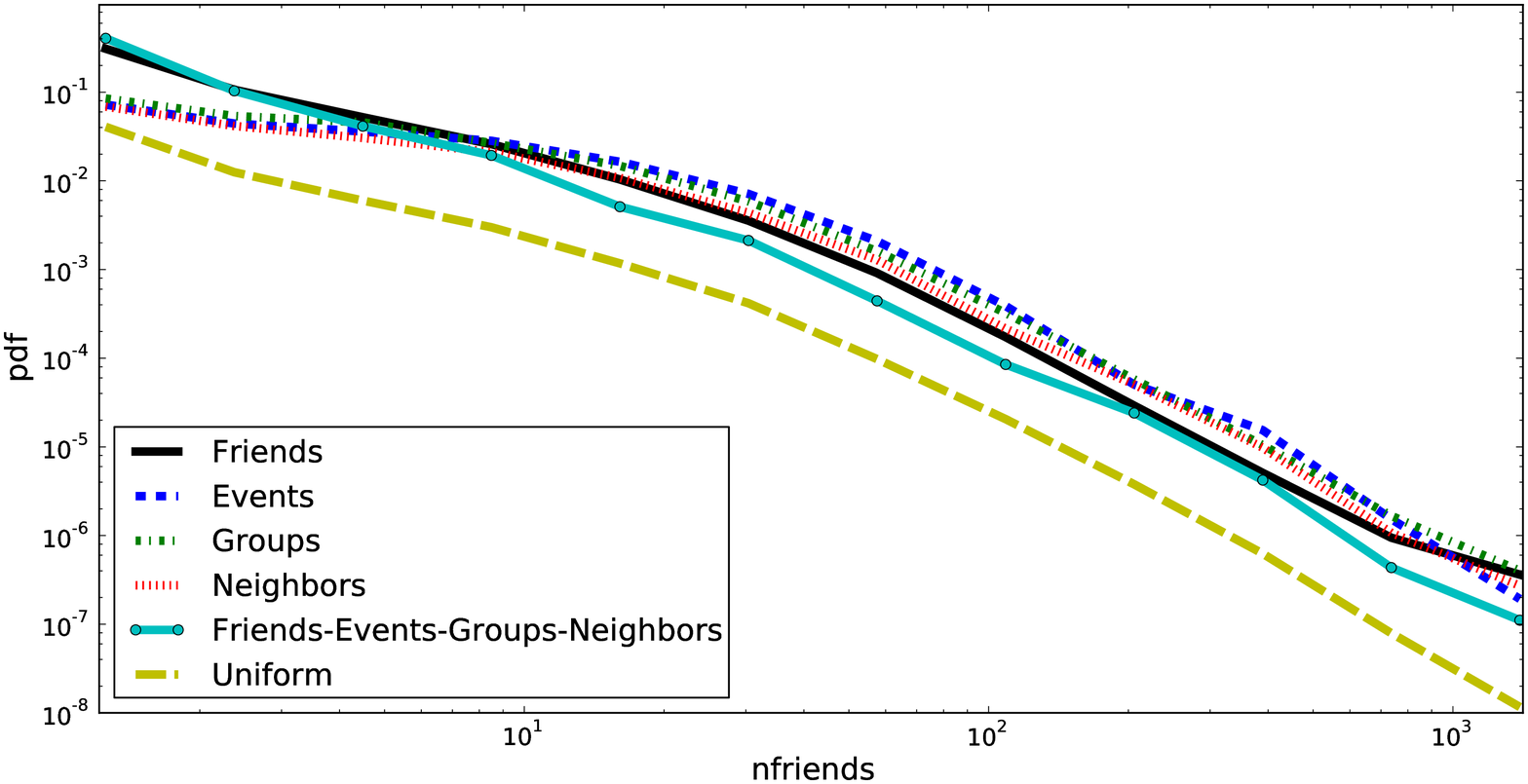}}
\subfigure[Number of groups]{\label{fig:pdf_ngroups}
\includegraphics[scale=0.23, angle=0]{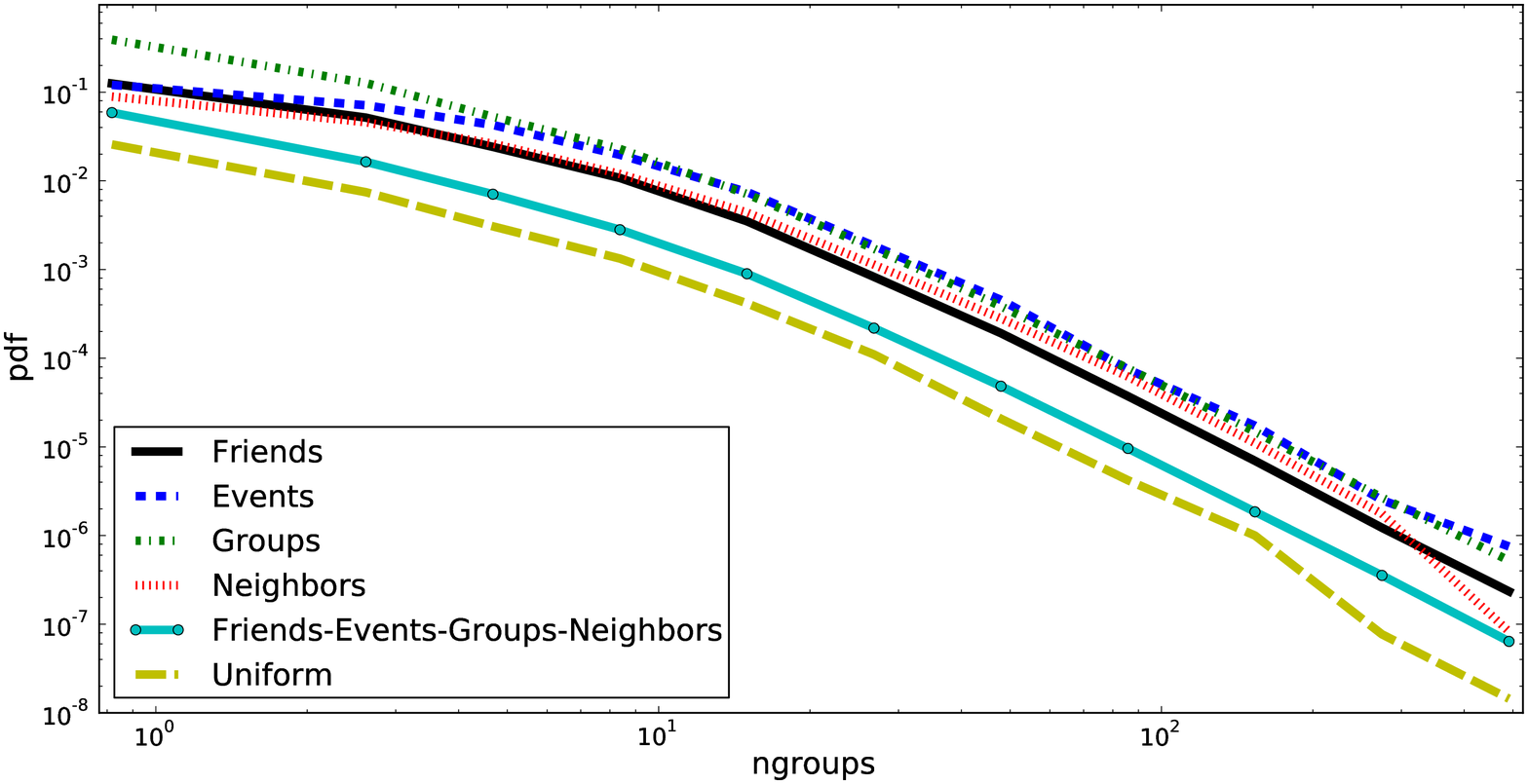}}
\subfigure[Number of past events ]{\label{fig:pdf_npast_events}
\includegraphics[scale=0.23, angle=0]{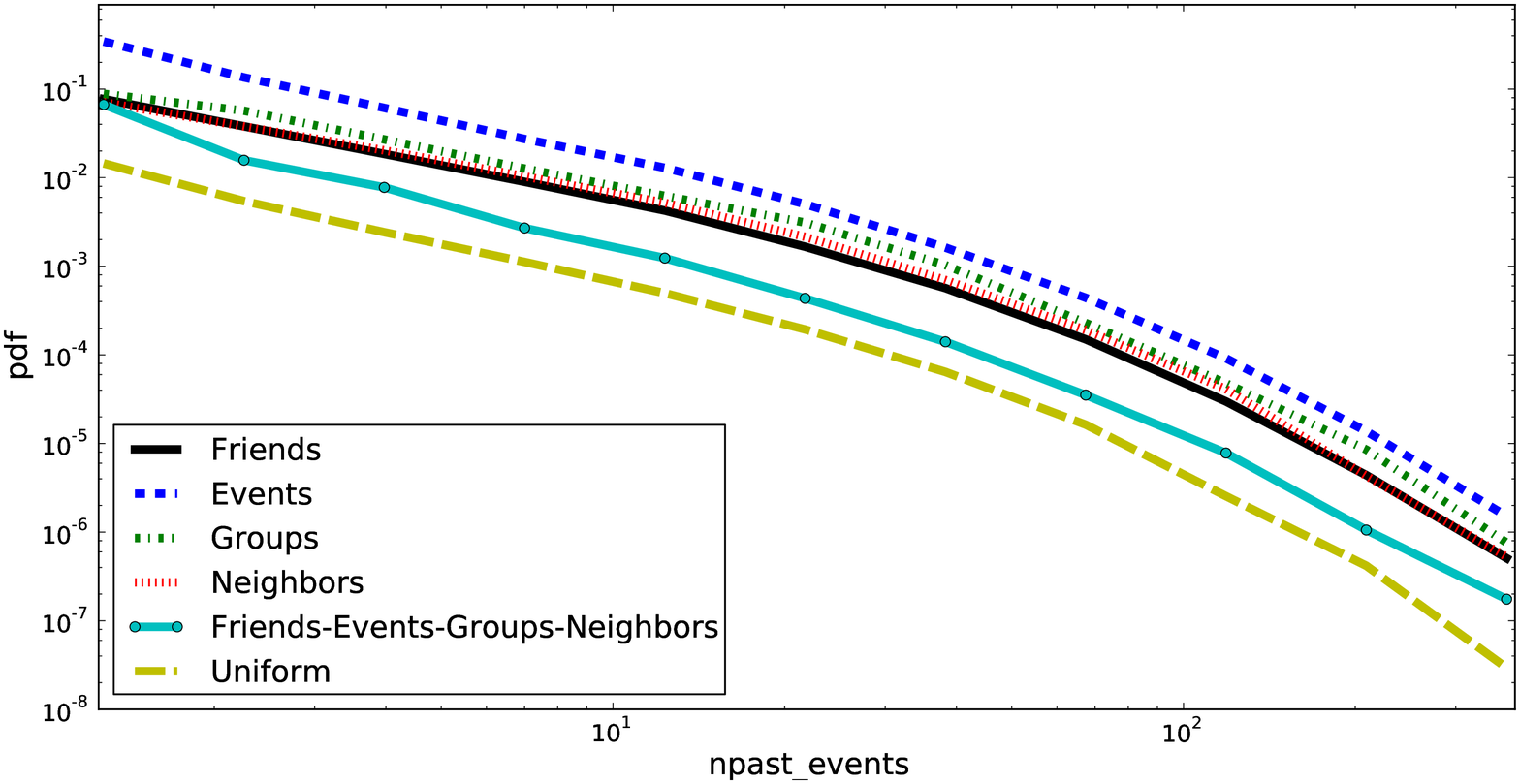}}
\subfigure[Number of future events ]{\label{fig:pdf_nfuture_events}
\includegraphics[scale=0.23, angle=0]{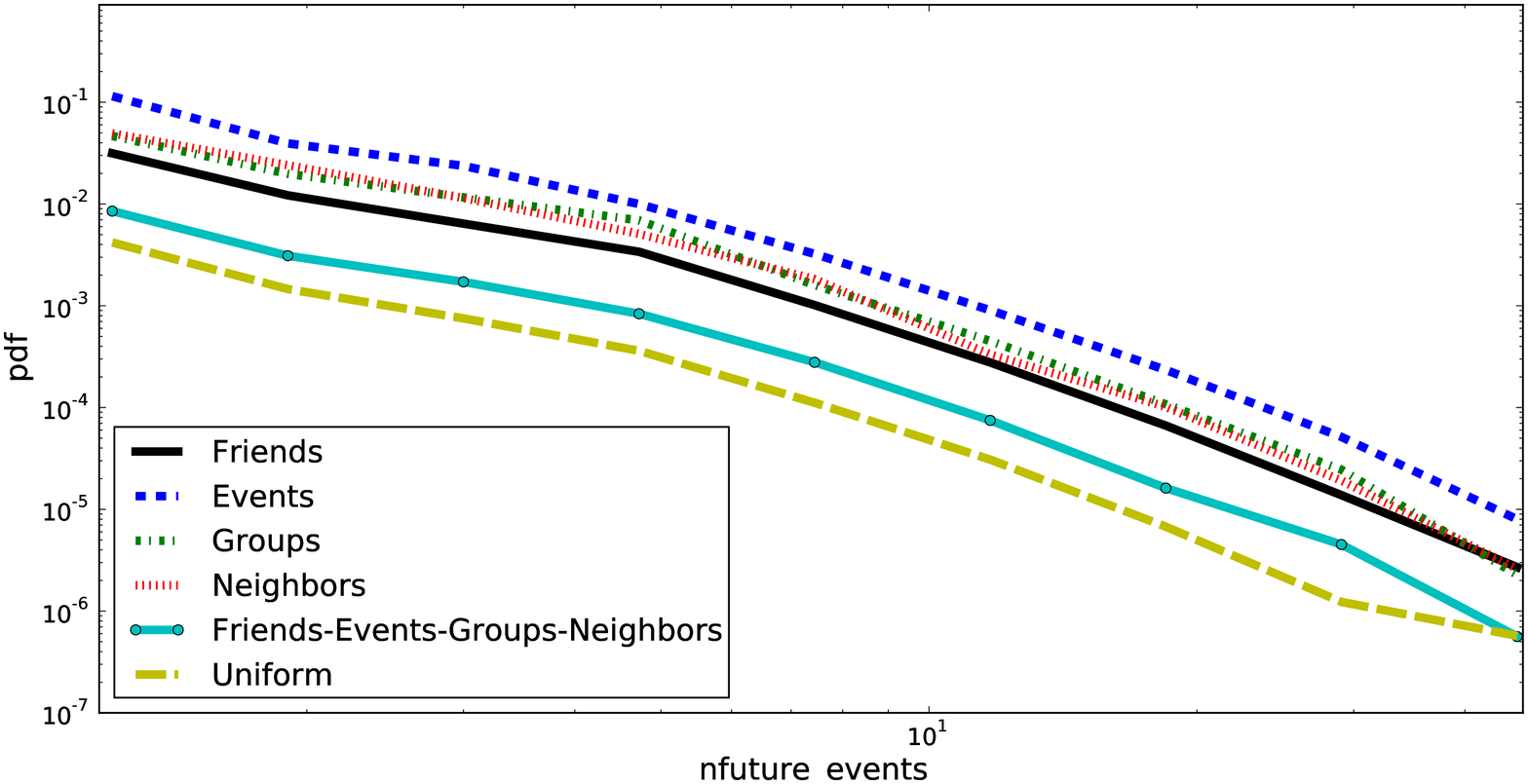}}
\caption{Single graph vs multigraph sampling. Probability distribution function (pdf) for  three  user properties (number of friends, number of groups, number of past events), as estimated by different crawl types.}
\label{fig:pdf_property}
\end{figure*}

\subsubsection{Comparing Samples to Ground Truth}
\label{sec:mg_comparegroundtruth}
~\\
{\em I. Comparing to UNI.}  In Table \ref{tab:percentage_isolates}, we saw that multigraph sampling was able to better approximate the percentage of isolates in the population.
 Here we consider other user properties (namely number of friends, past events, and groups a user belongs to, and whether he or she is a subscriber to \lastfm). In Fig \ref{fig:running_means}, we plot the sample mean value for four user properties across iteration number, for all crawl types and for the ground truth (UNI). One can see that crawling on a single graph, \ie Friends, Events, Groups, or Neighbors alone, leads to poor estimates.  This is prefigured by the previous results, as single graph crawls undersample individuals such as isolates on their corresponding relation, who form a large portion of the population. We also notice that Events and Groups alone consistently overestimate the averages, as these tend to cover the most active portion of the user base. However combining them together with other relations helps considerably. The multigraph that utilizes all relations, Friends-Events-Groups-Neighbors, is the closest to the truth. For example, it approximates very closely the avg number of groups and \% of paid subscribers (Figs \ref{fig:mean_ngroups}, \ref{fig:mean_subscriber}).

In Fig~\ref{fig:pdf_property}, we plot the probability distributions for four user properties of interest. Again, the crawl type Friends-Events-Groups-Neighbors is closest to the ground truth, in terms of shape of the distribution and vertical distance to it. Nevertheless, we observe that in both the probability distribution and the running mean plots, there is sometimes a gap from UNI, which is caused by the imperfect approximation of the \% of isolates. That is the reason that the gap is the largest for the number of friends property (Fig~\ref{fig:mean_nfriends}, \ref{fig:pdf_nfriends}).

{\em II. Comparing to Weekly Charts.} Finally, we compare the estimates obtained by different crawl types, to a different source of the ground truth - the weekly charts posted by \lastfm. This is useful as an example of how one can (at least approximately)
validate the representativeness of a random walk sample in the absence of a known uniform reference sample.

\begin{figure}[h]
\centering
\subfigure[Popularity of artists]{
\includegraphics[scale=0.23, angle=0]{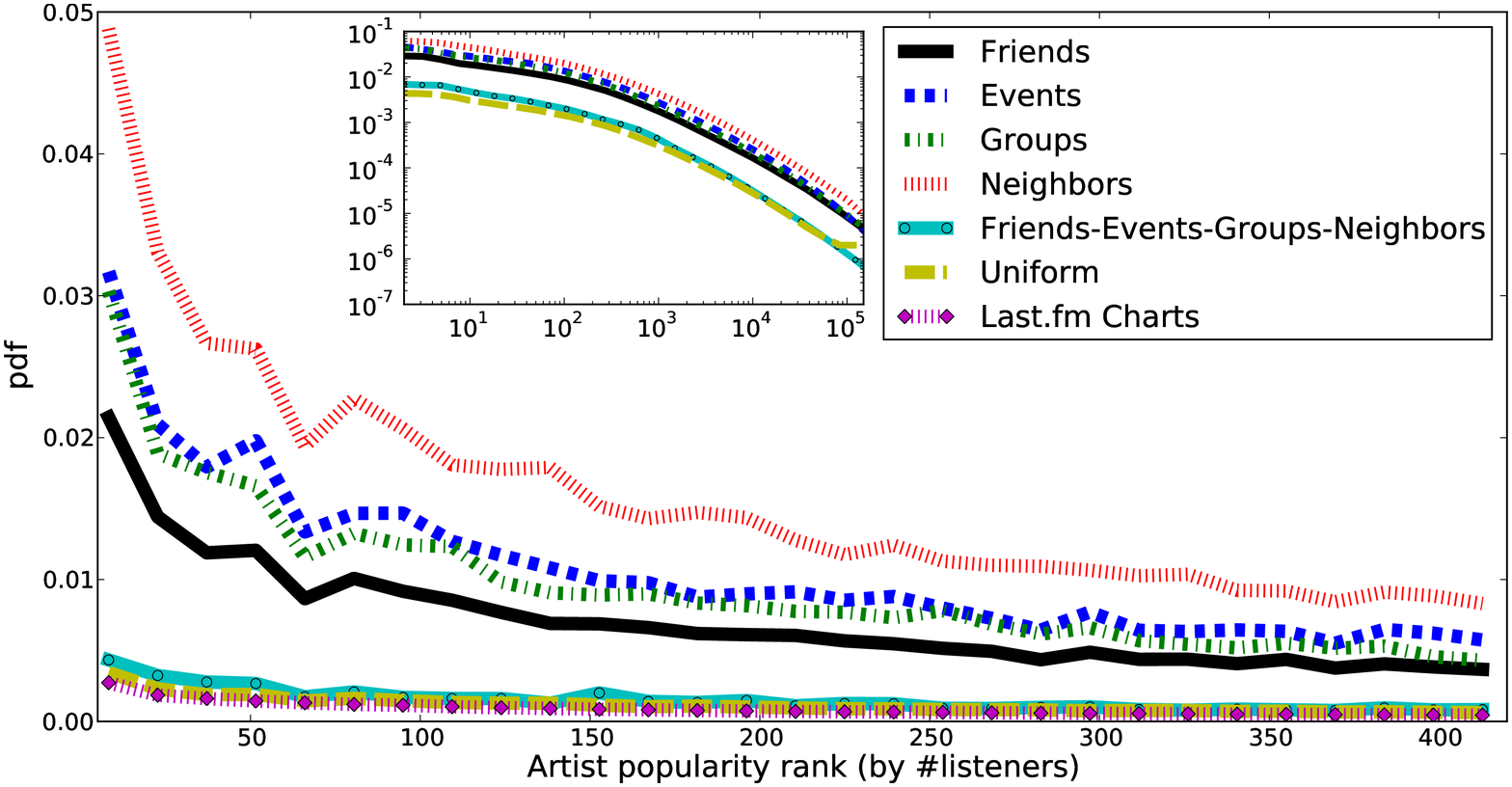}}
\subfigure[Popularity of tracks]{
\includegraphics[scale=0.23, angle=0]{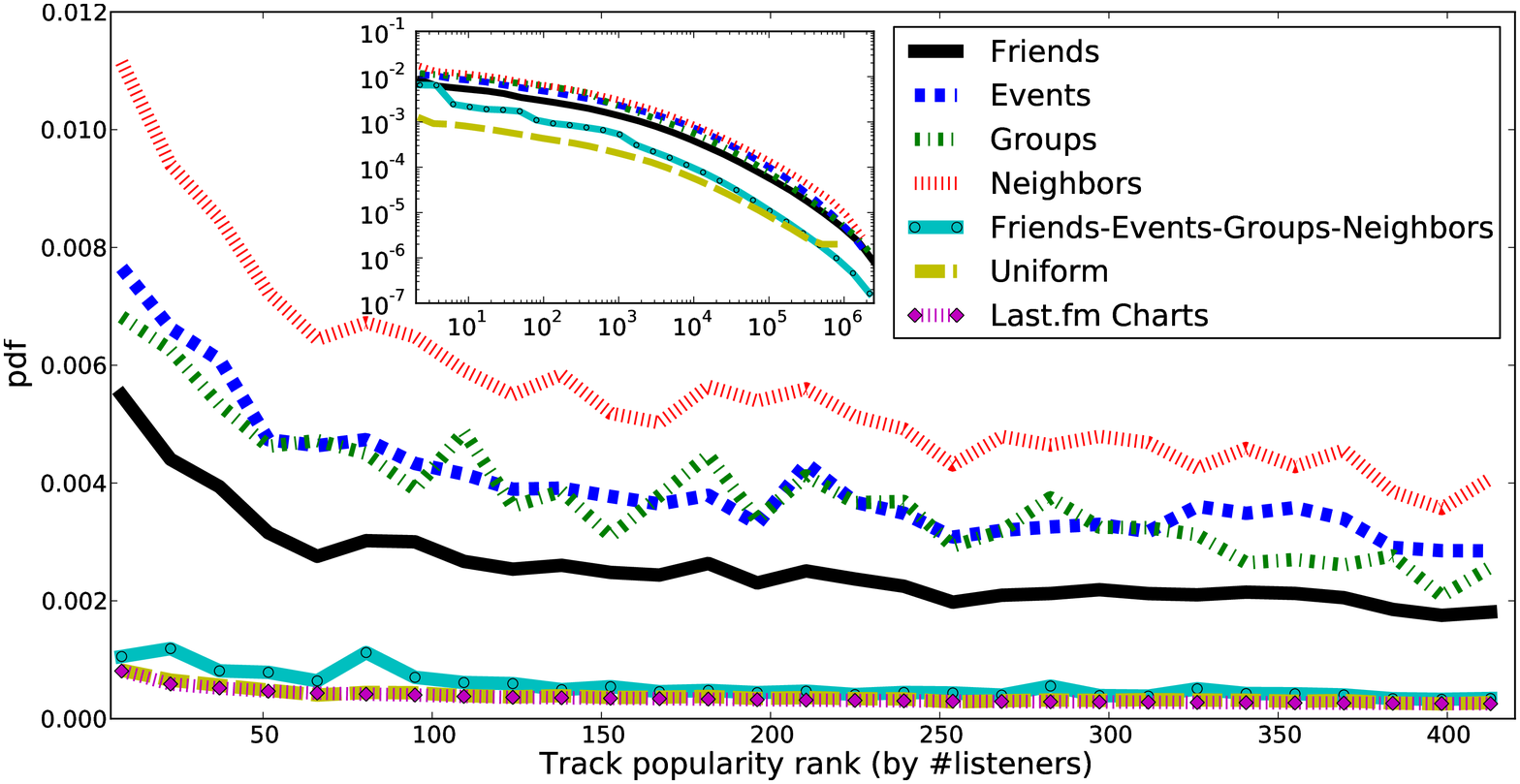}}
\caption{Weekly Charts
for the week 07/04-07/11. Artists/tracks are retrieved from ``\lastfm Charts`` and remain the same for all crawl types. Data is linearly binned (30 points). Inset: Artist/track popularity rank and percentage of listeners for all artists/tracks encountered in each crawl type.
}
\label{fig:comparison_charts}
\end{figure}

\lastfm reports on its website weekly music charts and statistics, generated automatically from user activity. To mention a few examples, ``Weekly Top Artists'' and ``Weekly Top Tracks'' as well as ``Top Tags``, ``Loved Tracks`` are reported. Each chart is based on the actual number of people listening to the track, album or artist recorded either through an Audioscrobbler plug-in (a free tracking service provided by the site) or the \lastfm radio stream. To validate the performance of multigraph sampling, we estimate the charts of ``Weekly Top Artists'' and ``Weekly Top Tracks'' from our sample of users for each of the crawl types in Table \ref{tab:datasets}, and we compare it to the published charts for the week July 04-July 11 2010, \ie the week just before the crawling started. To generate the charts from our user samples, we utilize API functions that allow us to fetch the exact list of artists and tracks that a user listened during a given date range. Fig. \ref{fig:comparison_charts} shows the observed artist/track popularity rank and the percentage of listeners for the top 420 tracks/artists (the maximum available) from the \lastfm Charts, with the estimated ranks and percentage of listeners for the same tracks/artists in each crawl type. As can be seen, the rank curve estimated from the multigraph Friends-Events-Groups-Neighbors tracks quite well the actual rank curve. Additionally, the curve that corresponds to the UNI sample is virtually lying on top of the ``\lastfm Charts'' line. On the other hand, the single graph crawl types Friends, Events, Groups, and Neighbors are quite far from actual charts. Here, as elsewhere, combining multiple relations gets us much closer to the truth than would reliance on a single graph.

\section{Related Work} \label{sec:mg_background}

Early graph exploration methods that were used to measure OSNs were based on BFS and snowball sampling \cite{Ahn-WWW-07, Mislove2007, Wilson09}. These methods have been shown to have a generally unknown bias towards high degree nodes when far from completion. In our recent and ongoing work, we attempt to correct for this bias \cite{Kurant2010, Kurant2011_JSAC_BFS}; however, BFS is out of the scope of this paper. Recent work in \cite{Rasti2008, Twitter08, Gjoka2010} used  random walks (where the bias is known) to sample users in OSNs, namely Friendster, Twitter and Facebook. Random walks have also been used to sample peer-to-peer networks \cite{Gkantsidis2004, Rasti09-RDS, Stutzbach2006-unbiased-p2p} and other large graphs \cite{Leskovec2006_sampling_from_large_graphs}.

Design of random walk techniques to improve mixing include \cite{Boyd2004_mixing, Ribeiro2010, GilksRoberts96, Thompson1990}.  Boyd \etal \cite{Boyd2004_mixing} pose the problem of finding the fastest mixing Markov Chain on a known graph as an optimization problem. However, in our case such an exact optimization is not possible since we are exploring an unknown graph. Ribeiro \etal \cite{Ribeiro2010} introduce Frontier sampling and explore multiple dependent random walks to improve sampling in disconnected or loosely connected subgraphs. Multigraph sampling  has the same goal but instead achieves it by exploring the social graph using multiple relations. Therefore, Frontier sampling is an orthogonal idea, which can potentially be combined with multigraph sampling for additional benefits. Multigraph sampling is also remotely related to techniques in the MCMC literature (\eg Metropolis-coupled MCMC or simulated tempering \cite{GilksRoberts96}) that seek to improve Markov chain convergence by mixing states across multiple chains with distinct stationary distributions. In \cite{Thompson1990, Thompson1991a} Thompson \etal introduce a family of adaptive cluster sampling (ACS) schemes, which are designed to explore nodes that satisfy some condition of interest; although random walk sampling is distinct from cluster sampling, the former does fit more broadly within the area of adaptive designs.

As noted in Section \ref{sec:mg_timescale}, we consider that the network of interest remains static during the duration of the crawl. We confirmed that this is a good approximation in the case of \lastfm by comparing two snapshots taken one week apart. Therefore, in this work, we do not consider dynamics, which are essential in other contexts \cite{Willinger09-OSN_Research, Rasti2008, acer2010random}.

Recent data collection studies of \lastfm include:
\cite{konstas2009social}, which develops a track recommendation system using social tags and friendship between users, \cite{schifanella2010folks}, which examines user similarity to predict social links,
and  \cite{baym2009tunes}, which explores the meaning of friendship in {\tt Last.fm} through survey sampling.
We emphasize that the importance of a representative sample is crucial to the usefulness of such datasets.

In our previous work \cite{Gjoka2010} and its extended version \cite{Gjoka2011_Facebook_JSAC}, we proposed a framework for crawling a single graph. %
In the implementation part of this paper, we adopt some of the practical recommendations of that work (\eg the use of the RWRW as the preferred crawling technique, the use of online convergence diagnostics, etc).
However, our focus here is on comparing multigraph sampling vs. single graph sampling, and on demonstrating its utility on fragmented networks such as \lastfm. 
To the best of our knowledge, our work is the first to explore sampling OSNs on a combination of multiple relations. 

\section{Conclusion}
\label{sec:mg_conclude}

In this paper, we have introduced {\em multigraph sampling} - a novel technique for random walk sampling of OSNs using multiple underlying relations. Multigraph sampling generates probability samples in the same manner as conventional random walk methods, but is more robust to poor connectivity and clustering within individual relations.  As we demonstrate using the \lastfm service, multigraph methods can give reasonable approximations to uniform sampling even where the overwhelming majority of users in each underlying relation are isolates, thus making single-graph methods fail.  Our experiments with synthetic graphs also suggest that multigraph sampling can improve  the coverage and the  convergence time for  partitioned or highly clustered networks.
Given these advantages, we believe multigraph sampling to be a useful addition to the growing suite of methods for sampling OSNs.

The focus of this paper was on (i) demonstrating the utility of multigraph sampling compared to singe graph sampling and (ii) on the design of a two-stage efficient algorithm that implements the idea.

Open questions include the selection of a few -out of many candidate- relations to use when sampling, so as to optimize the multigraph sampler performance.  Intuitively, we expect that negatively correlated relations will prove most effective. A related question is the weighting of the different relations for the same purpose. Gaining intuition into these problems will be particularly helpful in designing optimal OSN sampling schemes and is a direction for future work.

{
\bibliographystyle{IEEEtran}
\bibliography{osn-multigraph}
}

\vfill
\balance

\end{document}